\documentclass{article}
\usepackage{graphicx}
\usepackage{amsmath,amssymb,mathtools}
\usepackage{booktabs}
\usepackage{makecell}
\usepackage{url}
\usepackage{mathptmx}
\usepackage{times}
\makeatletter
\let\cl@part\relax  
\makeatother
\usepackage{algorithm}        
\usepackage{algpseudocode}    
\makeatletter
\@ifpackageloaded{algpseudocode}{%
  \providecommand{\STATE}{\State}
  \providecommand{\FOR}{\For}
  \providecommand{\ENDFOR}{\EndFor}

}{}
\makeatother
\newenvironment{proof}[1][Proof]{\par\noindent\textbf{#1.}\ }{\hfill$\square$\par}
\usepackage[numbers,sort&compress]{natbib}
\usepackage{aliascnt}
\makeatletter
\def\@begintheorem#1#2{%
  \trivlist
  \item[\hskip\labelsep{\bfseries #1\ #2}]\itshape}
\def\@opargbegintheorem#1#2#3{%
  \trivlist
  \item[\hskip\labelsep{\bfseries #1\ #2\ (#3)}]\itshape}
\makeatother
\newtheorem{theorem}{Theorem}

\newaliascnt{lemma}{theorem}
\newtheorem{lemma}[lemma]{Lemma}
\aliascntresetthe{lemma}

\newaliascnt{proposition}{theorem}
\newtheorem{proposition}[proposition]{Proposition}
\aliascntresetthe{proposition}

\newaliascnt{corollary}{theorem}
\newtheorem{corollary}[corollary]{Corollary}
\aliascntresetthe{corollary}

\newaliascnt{conjecture}{theorem}

\aliascntresetthe{conjecture}

\newaliascnt{definition}{theorem}
\newtheorem{definition}[definition]{Definition}
\aliascntresetthe{definition}

\newaliascnt{remark}{theorem}
\newtheorem{remark}[remark]{Remark}
\aliascntresetthe{remark}

\newaliascnt{assumption}{theorem}
\newtheorem{assumption}[assumption]{Assumption}
\aliascntresetthe{assumption}

\newaliascnt{hypothesis}{theorem}

\aliascntresetthe{hypothesis}

\newaliascnt{property}{theorem}

\aliascntresetthe{property}

\usepackage[capitalize]{cleveref}
\Crefname{figure}{Fig.}{Figs.}
\crefname{lemma}{Lemma}{Lemmas}
\Crefname{lemma}{Lemma}{Lemmas}
\crefname{proposition}{Proposition}{Propositions}
\Crefname{proposition}{Proposition}{Propositions}
\crefname{corollary}{Corollary}{Corollaries}
\Crefname{corollary}{Corollary}{Corollaries}
\crefname{definition}{Definition}{Definitions}
\Crefname{definition}{Definition}{Definitions}
\crefname{remark}{Remark}{Remarks}
\Crefname{remark}{Remark}{Remarks}
\crefname{assumption}{Assumption}{Assumptions}
\Crefname{assumption}{Assumption}{Assumptions}
\crefname{conjecture}{Conjecture}{Conjectures}
\Crefname{conjecture}{Conjecture}{Conjectures}
\crefname{hypothesis}{Hypothesis}{Hypotheses}
\Crefname{hypothesis}{Hypothesis}{Hypotheses}
\crefname{property}{Property}{Properties}
\Crefname{property}{Property}{Properties}

\usepackage{mpcsymbols}

\crefformat{equation}{(#2#1#3)}
\providecommand{\runauthor}[1]{}%
\providecommand{\runtitle}[1]{}%

\begin{document}

\title{Disturbance Attenuation Regulator I-A: Signal Bound Finite Horizon Solution\thanks{The authors gratefully acknowledge the financial support of the National Science Foundation (NSF) under Grant Nos. 2027091 and 2138985. The authors thank Moritz Diehl for helpful discussions.}}

\author{Davide Mannini\thanks{Department of Chemical Engineering, University of California, Santa Barbara. Email: dmannini@ucsb.edu} \and James B. Rawlings\thanks{Department of Chemical Engineering, University of California, Santa Barbara. Email: jbraw@ucsb.edu}}

\maketitle

\begin{abstract}
This paper develops a generalized finite horizon recursive solution to the discrete time signal bound disturbance attenuation regulator (SiDAR) for state feedback control. This problem addresses linear dynamical systems subject to signal bound disturbances, i.e., disturbance sequences whose squared signal two-norm is bounded by a fixed budget. The term generalized indicates that the results accommodate arbitrary initial states. By combining game theory and dynamic programming, we derive a recursive solution for the optimal state feedback policy valid for arbitrary initial states. The optimal policy is nonlinear in the state and requires solving a tractable convex scalar optimization for the Lagrange multiplier at each stage; the control is then explicit. For fixed disturbance budget $\alpha$, the state space partitions into two distinct regions: $\mathcal{X}_L(\alpha)$, where the optimal control policy is linear and coincides with the standard linear $H_{\infty}$ state feedback control, and $\mathcal{X}_{NL}(\alpha)$, where the optimal control policy is nonlinear. We establish monotonicity and boundedness of the associated Riccati recursions and characterize the geometry of the solution regions. A numerical example illustrates the theoretical properties.

This work provides a complete feedback solution to the finite horizon SiDAR for arbitrary initial states. Companion papers address the steady-state problem and convergence properties for the signal bound case, and the stage bound disturbance attenuation regulator (StDAR).
\end{abstract}

\section{Introduction}
\label{sec:intro}
The disturbance attenuation regulator (DAR), also known as the sensitivity minimization problem, is a deterministic game-theoretic robust control for systems affected by exogenous bounded disturbances.  In this framework, the control design seeks to ensure that the closed-loop system maintains good performance, i.e., low cost, despite any unknown but bounded disturbance. Specifically, the problem is structured as a sequential dynamic noncooperative zero-sum game (a Stackelberg game), i.e., a minmax optimization, where the disturbance (follower) optimizes first and the control (leader) optimizes second. Notably, solutions to such games need not satisfy strong duality nor correspond to stationary points, i.e., points in the domain of a function at which the gradient is zero.

The DAR has been formulated in two forms: the signal bound disturbance attenuation regulator (SiDAR), which constrains disturbances through a single bound over the entire time horizon, and the stage bound disturbance attenuation regulator (StDAR), which constrains disturbances independently at each time step.

The intellectual ancestor of the DAR is Bulgakov's \emph{disturbance accumulation problem} \citep{bulgakov:1946} from the 1940s, which asked for the maximal terminal state deviation under stagewise input bounds.  
Although not a game, Bulgakov's constrained maximization, widely studied in the Soviet/Russian literature and popularly known as the \textit{Bulgakov's problem} to these days, laid the groundwork for later game-theoretic stage bound disturbance attenuation treatments in the Soviet/Russian literature, but remained mostly ignored or unknown in the West.

The first game versions of the DAR appeared in the early 1960s: Gadzhiev \citep{gadzhiev:1962}, who obtained a nonlinear control policy for linear systems, treated the signal bound case, while Stikhin \citep{stikhin:1963} addressed the stage bound case.  
Dorato and Drenick subsequently introduced these ideas to the Western community \citep{dorato:drenick:1966}. Despite an intense burst of largely independent activity in both Eastern and Western research communities in the 1960s and 1970s \citep{koivuniemi:1966,ragade:sarma:1967,salmon:1968,rhodes:luenberger:1969,kimura:1970,medanic:andjelic:1971,ulanov:1971,bertsekas:rhodes:1973,yakubovich:1975,barabanov:granichin:1984}, progress soon stalled: even for linear systems the resulting minmax optimizations exhibit only weak duality for cases of interest. Standard gradient based optimization algorithms face fundamental difficulties in these problems because domain restrictions during iterative search, such as trust regions, may inadvertently exclude solution branches, preventing convergence to the global optimum even when part of the solution lies within the search region, as Witsenhausen observed \citep{witsenhausen:1968}. Consequently, a complete solution for either disturbance model remained elusive.

Interest reignited in the 1980s when Zames cast the signal bound problem in the frequency domain as the $H_{\infty}$ norm minimization of a disturbance to output transfer function matrix \citep{zames:1981}.  
Glover and Doyle \citep{glover:doyle:1988} translated that formulation back to the time domain, deriving dual Riccati recursions for continuous time, output feedback systems with zero initial state.  Basar \citep{basar:1989a} subsequently provided a finite and infinite horizon recursive dynamic game derivation (again assuming zero initial state), and Vidyasagar \citep{vidyasagar:1986} extended the framework to stage bound disturbances.  

Zames' frequency domain problem introduced the disturbance attenuation level $\gamma$ (a Lagrange multiplier analog), a device largely absent from the 1960s–1970s game-theoretic literature, but it also steered subsequent research toward zero initial state settings.  As a result, insights from the earlier game-theoretic line of work remained only weakly connected to the emerging $H_{\infty}$ theory. For historical accounts the reader may consult Dorato's review, which traces the development mainly through Western contributions~\citep{dorato:1987}, and the review by Khlebnikov, Polyak, and Kuntsevich, which focuses on the Soviet/Russian literature while still summarizing key Western results~\citep{khlebnikov:polyak:kuntsevich:2011}.

Didinsky and Basar \citep{didinsky:basar:1992} partially addressed nonzero initial states for the signal bound case by partitioning the state space into distinct solution regions, though their analysis relied on an auxiliary strongly dual reformulation and did not yield explicit solutions for the optimal control in all regions of the state space.  Khargonekar et al.~\citep{khargonekar:nagpal:poolla:1991} and Balandin et al.~\citep{balandin:kogan:krivdina:fedyukov:2014} considered nonzero initial states by introducing constraints combining disturbance norms with weighted quadratic functions of the initial state, but their formulations treat the initial state as a measured quantity rather than an uncertain parameter. Consequently, the resulting control law takes the form of linear state feedback in which the initial state is a measured or known quantity rather than an uncertain parameter.

To date there is \emph{no} direct feedback solution to the SiDAR that accommodates arbitrary initial states without appealing to auxiliary strongly dual problems. These gaps matter in practice: large setpoint changes or disturbances drive the system far from equilibrium where existing linear $H_{\infty}$ control is valid.

We close these gaps by deriving a generalized, finite horizon, recursive state feedback solution via dynamic programming for the SiDAR that:
\begin{itemize}
\item is valid for any initial condition
\item yields an optimal state feedback policy that is nonlinear in the state and requires solving a tractable convex scalar optimization for the Lagrange multiplier at each stage; the control gain is then explicit
\item reveals two qualitatively different regions: $\mathcal{X}_L(\alpha)$ where the optimal policy is linear and coincides with standard $H_{\infty}$ feedback, and $\mathcal{X}_{NL}(\alpha)$ where the policy is nonlinear
\end{itemize}
We expand the theory by proving monotonicity and boundedness of the Riccati recursion, characterizing the region geometry as ellipsoids centered at the origin, and establishing the derivative of the value function with respect to the Lagrange multiplier.

We pose the finite horizon SiDAR in \Cref{sec:setup}. \Cref{sec:signal} develops the solution and establishes monotonicity properties. The geometry of the solution regions is analyzed in \Cref{sec:control-region}. \Cref{sec:num-example} illustrates the theory with a numerical example, and \Cref{sec:end} summarizes the main findings. The appendix compiles fundamental propositions used throughout.

Companion papers address the steady-state problem and convergence properties for the signal bound case \citep{mannini:rawlings:2026b}, and the stage bound disturbance attenuation regulator (StDAR) \citep{mannini:rawlings:2026c}.

\textit{Notation:} Let $\bbR $ denote the reals and $\bbI$ the integers. $\mathbb{R}^{m \times n}$ denotes the space of $m \times n$ real matrices and $\mathbb{S}^n$ denotes the space of $n \times n$ real symmetric positive definite matrices.
The \(\norm{x}\) denotes the two-norm of vector \(x\); \(\norm{M}\) denotes the induced two-norm of matrix \(M\); \(\normf{M}_F\) denotes the Frobenius norm of matrix \(M\). For matrices $X, Y \in \mathbb{R}^{m \times n}$, the matrix inner product is $\langle X, Y \rangle \eqbyd \operatorname{Tr}(X'Y)$, and $\normf{M}_F = \sqrt{\langle M, M \rangle}$. For a symmetric matrix $A \in \mathbb{R}^{n \times n}$ with $A \succeq 0$, $A^{1/2}$ denotes the unique principal symmetric square root: $A^{1/2} \succeq 0$ and $(A^{1/2})^2 = A$. For $A \succ 0$, define $A^{-1/2} \eqbyd (A^{1/2})^{-1}$. For a symmetric matrix $\Gamma \succeq 0$, we may write $\Gamma = WW'$ where $W \eqbyd \Gamma^{1/2}$ denotes the principal square root unless stated otherwise; in general, such factorizations are not unique. For a vector $y \in \mathbb{R}^p$, let $\yseq$ denote a sequence $\yseq \eqbyd (y_0, y_1, \dots, y_{N-1})$. The two-norm of a signal $\yseq$ is defined as $\smax{\yseq} \eqbyd ( \sum_{k=0}^{N-1} \norm{y_k}^2 )^{1/2}$ for finite horizon and $\smax{\yseq} \eqbyd ( \sum_{k=0}^{\infty} \norm{y_k}^2 )^{1/2}$ for infinite horizon; the one-norm of a sequence is defined as $\smax{\yseq}_1 \eqbyd \sum_{k=0}^{N-1} \norm{y_k}$. The column space (range) and nullspace of a matrix $M$ are denoted by $\mathcal{R}(M)$ and $\mathcal{N}(M)$, respectively. The pseudoinverse of a matrix $M$ is denoted as $M^{\dagger}$. The determinant of a square matrix $M$ is denoted $\det M$, and the adjugate (classical adjoint) is denoted $\mathrm{adj}(M)$. For symmetric matrices $A$ and $B$, $A \succeq B$ denotes $A - B$ is positive semidefinite (the Loewner order); a minimal solution refers to the smallest solution in the Loewner order.
\section{SiDAR Set Up}
\label{sec:setup}
Consider the following discrete time system
\begin{equation}
    x^+ = Ax + Bu + Gw \label{system}
\end{equation}  
in which $x \in \bbR^n$ is the state, $u \in \bbR^m$ is the control, $w \in \bbW \subset \bbR^q$ is a disturbance, and $x^+\in \bbR^n$ is the successor state. Denote the horizon length, i.e., number of time steps in the horizon, as $N \in \mathbb{I}_{\ge 1}$. Define the control and disturbance sequences: $\useq \eqbyd (u_0,u_1,\dots,u_{N-1})$, $\wseq \eqbyd (w_0,w_1,\dots,w_{N-1})$. Consider the following signal bound disturbance constraint set (signal two-norm bound)
\[
	\bbW
	\eqbyd
	\Bigl\{
	\wseq
	 \mid
	\sum_{k=0}^{N-1}|w_k|^2 \leq \alpha
	\Bigr\}
    \] 
Define the following objective function
\begin{equation}
    V(x_0, \useq, \wseq) = \sum_{k=0}^{N-1} \ell(x_k, u_k)  + \ell_f(x_N) \label{maincost}
\end{equation}
where $x_0$ is the initial state, $\ell(\cdot)$ the stage cost, $\ell_f(\cdot)$ the terminal cost
\[
    \ell(x,u) = (1/2)x'Qx + (1/2)u'Ru \qquad \ell_f(x) = (1/2)x'P_fx
\]
in which $Q \succeq 0$, $R \succ 0$, and $P_f \succeq 0$. The state $x_k$ in \eqref{maincost} denotes the solution of \eqref{system} at stage $k$ generated from the initial state $x_0$ by the control and disturbance sequences $\useq$ and $\wseq$. We make the following assumptions.
\begin{assumption}
For the linear system \eqref{system}, $(A,B)$ stabilizable and  $(A,Q)$ detectable. \label{asst1}
\end{assumption}
\begin{assumption}
$\mathcal{R}(G)\subseteq\mathcal{R}(B)$. \label{asst2}
\end{assumption}
Assumption~\ref{asst2} ensures every column of $G$ lies in the column space of $B$, which is a sufficient condition for nonsingularity of the block matrix appearing in the Riccati recursion of \cref{sec:signal} (\cref{prop:range-inclusion-invertible}). Relaxing it requires pseudoinverses in place of inverses and is left to future work.
\begin{assumption}
$G'P_fG\neq0$. \label{asst3}
\end{assumption}
\begin{assumption}
$Q\succ0$, $P_f \succ0$. \label{asst4}
\end{assumption}
The strict definiteness in Assumption~\ref{asst4} is invoked starting from \cref{prop:derivative}; it is not required for the dynamic programming solution in \cref{prop:1dsignal,prop:ndsignal}.

The \textit{signal bound disturbance attenuation regulator} (SiDAR) is obtained from the optimization problem
\begin{equation}
    V^*(x_0) \eqbyd \min_{u_0}\max_{w_0} \; \min_{u_1}\max_{w_1} \; \cdots \min_{u_{N-1}}\max_{w_{N-1}} \;
 \frac{V(x_0, \useq, \wseq)}{\sum^{N-1}_{k=0} |w_k|^2 } \quad \wseq \in \bbW\label{signaldp-w}
\end{equation}
subject to \eqref{system}. Under \cref{asst3}, the maximum over the disturbance is attained on the boundary $\sum_{k=0}^{N-1}|w_k|^2 = \alpha > 0$, as shown in the proof of \cref{prop:1dsignal}; the denominator of the ratio in \eqref{signaldp-w} is therefore strictly positive at the optimum and the ratio is well-defined.

\subsection{Dynamic Programming}

Although the constraint $\sum_{k=0}^{N-1}|w_k|^2 \leq \alpha$ couples disturbance choices across all $k$ stages, the SiDAR admits a standard Bellman recursion by augmenting the state with the remaining disturbance budget $b \in [0,\alpha]$ at each stage. Define the augmented value function $V_k: \bbR^n \times [0,\alpha] \to \bbR$ satisfying
\begin{equation}
\begin{split}
V_k(x,b) = \min_u \max_{|w|^2 \leq b} \Big[ & \frac{1}{\alpha}\ell(x,u) \\
& + V_{k+1}(Ax+Bu+Gw, b - |w|^2) \Big]
\end{split}
\label{bellman-recursion-signal}
\end{equation}
for $k \in \{0,\ldots,N-1\}$, with the boundary condition at $k=N$ given by
\begin{equation*}
V_N(x,b) \eqbyd V_N(x)= \frac{1}{\alpha} \ell_f(x) \qquad \text{for all $b$} 
\end{equation*}
The budget evolves as $$b_{k+1} = b_k - |w_k|^2$$ with initial condition $b_0 = \alpha$. The optimal control policy at stage $k$ is
\begin{equation}
\begin{split}
u_k^*(x,b) = \arg\min_u \max_{|w|^2 \leq b} \Big[ & \frac{1}{\alpha}\ell(x,u) \\
& + V_{k+1}(Ax+Bu+Gw, b - |w|^2) \Big]
\end{split}
\label{opt-control-policy-signal}
\end{equation}
Substituting $u_k^*(x,b)$ into \eqref{bellman-recursion-signal} yields
\begin{equation*}
\begin{split}
V_k(x,b) = \max_{|w|^2 \leq b} \Big[ & \frac{1}{\alpha}\ell(x,u_k^*(x,b)) \\
& + V_{k+1}(Ax+Bu_k^*(x,b)+Gw, b - |w|^2) \Big]
\end{split}
\end{equation*}
and the optimal disturbance policy is
\begin{equation*}
\begin{split}
w_k^*(x,b) = \arg\max_{|w|^2 \leq b} \Big[ & \frac{1}{\alpha}\ell(x,u_k^*(x,b)) \\
& + V_{k+1}(Ax+Bu_k^*(x,b)+Gw, b - |w|^2) \Big]
\end{split}
\end{equation*}
The optimal cost to \eqref{signaldp-w} is $V^*(x_0) = V_0(x_0, \alpha)$.

The inner maximization in \eqref{bellman-recursion-signal} is a constrained quadratic optimization over the compact set $\bbW_k(b) = \{w : |w|^2 \leq b\}$. The augmented state problem \eqref{bellman-recursion-signal} is intractable: the value function $V_k(x,b)$ depends on both the continuous state $x$ and continuous budget $b \in [0,\alpha]$, requiring representation of $V_k$ over a two dimensional continuum. Discretizing $b$ alone does not resolve this difficulty, as $V_k(x,b_i)$ must still be represented as a function of continuous $x$ for each budget level (or $x$ must also be discretized, yielding an $(n+1)$ dimensional grid). The backward recursion for $V_k$ and forward evolution of $b_k$ via $b_{k+1} = b_k - |w_k^*|^2$ present no fundamental obstacle (this is standard in dynamic programming) but the continuous representation does. Introducing a Lagrange multiplier $\lambda \geq 0$ for the aggregate budget constraint eliminates $b$ from the state: for any fixed $\lambda$, the problem admits tractable Riccati recursions in the original state $x$ alone, followed by a scalar convex optimization over $\lambda$ at each measured state, as derived in \cref{sec:signal}.


\begin{remark}
    Note that problem \eqref{signaldp-w} with $N \rightarrow \infty$ reduces to the standard $H_{\infty}$ robust control problem. Define
    \begin{equation*}
z = \frac{1}{\sqrt{2}}\begin{bmatrix}
Q^{1/2} & 0 \\
0 & R^{1/2}
\end{bmatrix}
\begin{bmatrix} x\\ u \end{bmatrix}
\end{equation*}
Then the objective function is
\[
\frac{V(x_0, \useq, \wseq)}{\sum^{\infty}_{k=0} |w_k|^2 } =  \frac{\smax{\zseq}^2}{\smax{\wseq}^2}
\]
and the SiDAR is equivalently expressed by
\[
\min_{\useq} \max_{\smax{\wseq}^2 = \alpha} \; \frac{\smax{\zseq}^2}{\smax{\wseq}^2}
\]
which is the standard time domain $H_{\infty}$ state feedback problem.
\end{remark}

\section{SiDAR Solution}
\label{sec:signal}
\subsection{Two-stage Solution}
We solve the two-stage version of the SiDAR \eqref{signaldp-w} for the linear system \eqref{system}. The signal bound constraint $\sum_{k=0}^{N-1}|w_k|^2 \leq \alpha$ couples the disturbances across stages; introducing a Lagrange multiplier $\lambda \geq 0$ for this constraint reduces the problem to a scalar convex optimization in $\lambda$ alone, with the inner stagewise minmax problems solved in closed form for each fixed $\lambda$. The two-stage case demonstrates how $\lambda$ is deferred from the inner stages to the outermost optimization at stage~$0$.

A two-stage SiDAR is
\begin{equation}
    V^*(x_0) \eqbyd \min_{u_0}\max_{w_0} \; \min_{u_1}\max_{w_1} \;
 \frac{V(x_0, \useq, \wseq)}{\sum^{1}_{k=0} |w_k|^2} \quad \wseq \in \bbW \label{1dsignaldp}
\end{equation}
where $\useq \eqbyd (u_0,u_1)$ and $\wseq \eqbyd (w_0,w_1)$ and the objective function is
\[
V(x_0, \useq, \wseq) = (1/2)\bigg( x_0'Qx_0  + u_0'Ru_0 + x_1'Qx_1 + u_1'Ru_1 + x_2P_fx_2\bigg)
\]

\Cref{prop:1dsignal} states the resulting solution: the optimal multiplier $\lambda^*(x_0)$ minimizes a scalar convex value function over a feasibility domain $[\lambda_1, \infty)$ that the proposition also constructs, and the optimal control and disturbance policies follow from the stationary conditions of the stacked Lagrangian evaluated at $\lambda^*$.

\begin{proposition}[Two-stage SiDAR \eqref{1dsignaldp}]
    \label{prop:1dsignal}
    Let Assumptions 1-3 hold. Consider the following scalar convex optimization 
\begin{gather}
\min_{\lambda \in [\lambda_1, \infty)} \; \frac12\!\left(\frac{x_0}{\sqrt{\alpha}}\right)'\! \Pi_0(\lambda)\!\left(\frac{x_0}{\sqrt{\alpha}}\right) +\frac{\lambda}{2} \label{eq:scalar-opt}\\ 
\lambda_2\eqbyd \norm{G'P_fG} \nonumber\\
\lambda_1 \eqbyd \begin{cases} \displaystyle \min_{\lambda\ge \lambda_2} \Bigl\{\, \lambda \;:\; \lambda = \norm{ G'\Pi_1(\lambda)G} \Bigr\} & \text{if }\;\norm{G'\Pi_1(\lambda_2)G}>\lambda_2\\[10pt] \lambda_2 & \text{if }\;\norm{G'\Pi_1(\lambda_2)G}\le\lambda_2 \end{cases} \nonumber\\ 
\Pi_0(\lambda) = Q+A'\Pi_1A -A' \Pi_1 \begin{bmatrix} B & G \end{bmatrix} M_0(\lambda)^{-1} \begin{bmatrix} B' \\ G'\end{bmatrix} \Pi_1A \nonumber\\ 
\Pi_1(\lambda) = Q+A'P_fA -A' P_f \begin{bmatrix} B & G \end{bmatrix} M_1(\lambda)^{-1} \begin{bmatrix} B' \\ G'\end{bmatrix} P_fA \nonumber
\end{gather}
where 
\begin{align*}
M_0(\lambda) &\eqbyd \begin{bmatrix}
        B'\Pi_1B+R & B'\Pi_1G \\
        G'\Pi_1B & G'\Pi_1G-\lambda I
    \end{bmatrix} \\
M_1(\lambda) &\eqbyd \begin{bmatrix}
        B'P_{f}B+R & B'P_{f}G \\
        G'P_{f}B & G'P_{f}G-\lambda I
    \end{bmatrix}
\end{align*}
Given the solution to the scalar convex optimization \eqref{eq:scalar-opt}, $\lambda^*(x_0)$, and terminal condition $P_f \succeq0$, then 
    \begin{enumerate}
        \item The optimal control policies $u^*_0(x_0;\lambda^*)$ and $u^*_1(x_1;\lambda^*)$ to \eqref{1dsignaldp} satisfies the stationary conditions
    \begin{align}
    M_0(\lambda^*)
    \begin{bmatrix} u_0 \\ z_0 \end{bmatrix}^* =- \begin{bmatrix} B' \\ G'\end{bmatrix} \Pi_1A \; x_0 \label{2dsc-a} \\
    M_1(\lambda^*)
    \begin{bmatrix} u_1 \\ z_1 \end{bmatrix}^* = -\begin{bmatrix} B' \\ G'\end{bmatrix} P_fA \; x_1 \label{2dsc-b}
    \end{align} 
    \item The optimal disturbance policies $w^*_0(x_0;\lambda^*) = \overline{w}_0(u^*_0) \cap \bbW$ and $w^*_1(x_1;\lambda^*) = \overline{w}_1(u^*_1) \cap \bbW$ to \eqref{1dsignaldp} satisfies
        \begin{align}
            (B'\Pi_1G)'u^*_0(x_0;\lambda^*)+(G'\Pi_1G - \lambda^* I) \ \overline{w}_0(u^*_0) = - G'\Pi_1Ax_0
            \label{2dwcond-a} \\
             (B'P_fG)'u^*_1(x_0;\lambda^*)+( G'P_fG - \lambda^* I ) \ \overline{w}_1(u^*_1) = - G'P_fAx_1
             \label{2dwcond-b} 
        \end{align}
    \item The optimal cost to \eqref{1dsignaldp}  is
    \begin{equation}
    V^*(x_0) = (1/2)\;(\frac{x_0}{\sqrt{\alpha}})' \Pi_0(\lambda^*) (\frac{x_0}{\sqrt{\alpha}})  + \lambda^* /2 \label{2doc}
    \end{equation} 
    \item For $\lambda\geq\lambda_1$, we have that $\Pi_{0}(\lambda) \succeq \Pi_{1}(\lambda)\succeq P_f$, and $\Pi_0(\lambda)$ and $\Pi_1(\lambda)$ are monotonic nonincreasing in $\lambda$.
   \end{enumerate}
\end{proposition}

\noindent\textbf{Proof sketch.}\ The proof has four blocks: boundary attainment of the disturbance optimum under \cref{asst3}, deferment of the Lagrange multiplier $\lambda$ to the outer optimization via the stacked Lagrangian and strong duality, construction of the feasibility bound $\lambda_1$ from the admissibility condition $\lambda \geq \norm{G'\Pi_1(\lambda)G}$, and joint convexity of the value function $L(\lambda)$ in $(\lambda, x_0)$. The full proof is in the appendix (\cref{app:proofs}).

\subsection{Finite Horizon Solution}
We now generalize to derive the recursive optimal solution to the finite horizon SiDAR \eqref{signaldp-w}
\begin{equation*}
    V^*(x_0) \eqbyd \min_{u_0}\max_{w_0} \; \min_{u_1}\max_{w_1} \; \cdots \min_{u_{N-1}}\max_{w_{N-1}} \;
 \frac{V(x_0, \useq, \wseq)}{\sum^{N-1}_{k=0} |w_k|^2 } \quad \wseq \in \bbW
\end{equation*}
where $\useq \eqbyd (u_0,u_1,\dots,u_{N-1})$, $\wseq \eqbyd (w_0,w_1,\dots,w_{N-1})$, and the objective function is \eqref{maincost}
\begin{equation*}
    V(x_0, \useq, \wseq) = \sum_{k=0}^{N-1} \ell(x_k, u_k)  + \ell_f(x_N)
\end{equation*}

\Cref{prop:ndsignal} extends \cref{prop:1dsignal} to horizon $N$: the optimal multiplier $\lambda^*(x_0)$ solves a scalar convex optimization on $[\lambda_1, \infty)$, and the optimal control, disturbance, and value at each stage $k$ are obtained from a backward Riccati recursion in $\Pi_k(\lambda)$. The state feedback policy is obtained by resolving the scalar optimization at each stage for the current state, as detailed in Remarks~11 and~12.

\begin{proposition}[Finite horizon SiDAR \eqref{signaldp-w}]
    \label{prop:ndsignal}
    Let Assumptions 1-3 hold. Consider the following scalar convex optimization 
\begin{gather}
\mathbf{L}_{si}: \quad \min_{\lambda \in [\lambda_1, \infty)}
     \; \frac12\!\left(\frac{x_0}{\sqrt{\alpha}}\right)'\!
           \Pi_0(\lambda)\!\left(\frac{x_0}{\sqrt{\alpha}}\right)
     +\frac{\lambda}{2}\ \label{lsi} 
\end{gather}
\begin{gather*}
\lambda_N \eqbyd \norm{G'P_fG} \\
\begin{split}
\lambda_k \eqbyd
\begin{cases}
 \begin{aligned}
  &\min_{\lambda\ge \lambda_{k+1}} \Bigl\{\, \lambda : \lambda =  \norm{ G'\Pi_{k+1}(\lambda)G} \Bigr\}
 \end{aligned} \\
 \qquad \qquad \text{if }\;\norm{G'\Pi_{k+1}(\lambda_{k+1})G}>\lambda_{k+1}\\[10pt]
 \lambda_{k+1} \\
 \qquad \qquad \text{if }\;\norm{G'\Pi_{k+1}(\lambda_{k+1})G}\le\lambda_{k+1}
\end{cases}
\end{split}
\end{gather*}
    subject to the Riccati recursion
    \begin{equation}
        \Pi_k(\lambda) = Q+A'\Pi_{k+1}A-A' \Pi_{k+1} \begin{bmatrix} B & G \end{bmatrix}
     M_k(\lambda)^{-1}
    \begin{bmatrix} B' \\ G'\end{bmatrix}\Pi_{k+1}A \label{signalrec1}
    \end{equation}
    where
    \[
    M_k(\lambda) \eqbyd \begin{bmatrix} B'\Pi_{k+1}B + R & B'\Pi_{k+1} G \\ (B'\Pi_{k+1}G)' & G'\Pi_{k+1}G - \lambda I  \end{bmatrix}
    \]
    for $k \in [0,1,\dots,N-1]$ and terminal condition $\Pi_N = P_f \succeq0$. Given the solution to the \textit{scalar} convex optimization \eqref{lsi}, $\lambda^*(x_0)$, then
    \begin{enumerate}
        \item The optimal control policy $u^*_k(x_k;\lambda^*)$ to \eqref{signaldp-w} satisfies the stationary conditions
    \begin{equation}
     M_k(\lambda^*)
    \begin{bmatrix} u_k \\ z_k \end{bmatrix}^* = -\begin{bmatrix} B' \\ G'\end{bmatrix} \Pi_{k+1}A \; x_k \label{ndsc2}
    \end{equation} 
    \item The optimal disturbance policy $w^*_k(x_k;\lambda^*) = \overline{w}_k(u^*_k) \cap \bbW$ to \eqref{signaldp-w} satisfies
    \begin{equation}
    \begin{split}
    (B'\Pi_{k+1}G)'u^*_k(x_k;\lambda^*) &+(G'\Pi_{k+1}G - \lambda^* I ) \ \overline{w}_k(u^*_k) \\
    &= -G'\Pi_{k+1}Ax_k
    \end{split}
    \label{ndwcond2}
    \end{equation}
    \item The optimal cost to \eqref{signaldp-w} is
    \begin{equation}
    V^*(x_0) = (1/2)\;(\frac{x_0}{\sqrt{\alpha}})' \Pi_0(\lambda^*) (\frac{x_0}{\sqrt{\alpha}})  + \lambda^* /2 \label{ndoc2}
    \end{equation} 
    \item For $\lambda\geq\lambda_1$, $\Pi_{k}(\lambda)$ is monotonic nonincreasing in $k$ and in $\lambda$.
    \end{enumerate}
\end{proposition}

\noindent\textbf{Proof sketch.}\ The proof extends \cref{prop:1dsignal} to horizon $N$ by induction. Each stage applies the same four-block argument (boundary attainment, stacked Lagrangian and deferment of $\lambda$, feasibility bound construction, joint convexity), with the bound $\lambda_k$ replacing $\lambda_1$ at stage $k$. The full proof is in the appendix (\cref{app:proofs}).

\Cref{prop:derivative} establishes three results: the derivative $dL/d\lambda$ in closed form (item~1), its monotonicity on $[\lambda_1, \infty)$ (item~2), and the condition $\lambda^* = \lambda_1$ if and only if $|\mathbf{z}^*(\lambda_1)|^2 \leq \alpha$ (item~3). Item~3 is the basis for the geometry of the solution regions in \cref{sec:control-region} (it expresses $\mathcal{X}_L(\alpha)$ and $\mathcal{X}_{NL}(\alpha)$ as ellipsoids in $\bbR^n$); items~1 and~3 are also used in the online \cref{alg:online-signal} to detect whether $\lambda^* = \lambda_1$ at a given measured state.

\begin{proposition}[Derivative of the SiDAR value function]
\label{prop:derivative}
Let Assumptions 1-4 hold. Consider the finite horizon SiDAR \eqref{signaldp-w} with value function
\begin{equation}
L(\lambda) \eqbyd \frac{1}{2}\left(\frac{x_0}{\sqrt{\alpha}}\right)' \Pi_0(\lambda) \left(\frac{x_0}{\sqrt{\alpha}}\right) + \frac{\lambda}{2}
\label{eq:L-lambda}
\end{equation}
defined for $\lambda \geq \lambda_1$, where $\lambda_1$ is the feasibility bound from \cref{prop:ndsignal}. Define the Lagrangian stationary point $\mathbf{z}^*(\lambda) \eqbyd \tilde{J}(\lambda) x_0$ where
\begin{equation}
\tilde{J}(\lambda) \eqbyd \begin{bmatrix}
J_0(\lambda) \\
J_1(\lambda) \Phi_{1,0}(\lambda) \\
\vdots \\
J_{N-1}(\lambda) \Phi_{N-1,0}(\lambda)
\end{bmatrix}
\label{eq:z-recursive}
\end{equation}
with
\begin{gather*}
M_k(\lambda) \eqbyd \begin{bmatrix} B'\Pi_{k+1}(\lambda) B+R & B'\Pi_{k+1}(\lambda) G \\ G'\Pi_{k+1}(\lambda) B & G'\Pi_{k+1}(\lambda) G-\lambda I \end{bmatrix} \\
d_k(\lambda) \eqbyd \begin{bmatrix} B'\Pi_{k+1}(\lambda) A \\ G'\Pi_{k+1}(\lambda) A \end{bmatrix}\\
K_k(\lambda) \eqbyd -\begin{bmatrix} I & 0 \end{bmatrix} M_k(\lambda)^{-1} d_k(\lambda) \\
J_k(\lambda) \eqbyd -\begin{bmatrix} 0 & I \end{bmatrix} M_k(\lambda)^{-1} d_k(\lambda)\\
F_k(\lambda) \eqbyd A + BK_k(\lambda) + GJ_k(\lambda) \\
\Phi_{k,j}(\lambda) \eqbyd F_{k-1}(\lambda)F_{k-2}(\lambda)\cdots F_j(\lambda)
\end{gather*}
for $j < k$ and $\Phi_{k,k}(\lambda) \eqbyd I$, and $\Pi_N(\lambda) = P_f$. Then
\begin{enumerate}
\item The derivative of $L(\lambda)$ for $\lambda \geq \lambda_1$ is
\begin{equation}
\frac{dL}{d\lambda} = \frac{1}{2} - \frac{1}{2}\frac{|\mathbf{z}^*(\lambda)|^2}{\alpha}
\label{eq:dL-dlambda}
\end{equation}
\item The derivative $dL/d\lambda$ is nondecreasing on $[\lambda_1, \infty)$.
\item The optimal multiplier satisfies $\lambda^* = \lambda_1$ if and only if $|\mathbf{z}^*(\lambda_1)|^2 \leq \alpha$.
\end{enumerate}
\end{proposition}

\noindent\textbf{Proof sketch.}\ The proof has three blocks. (i)~The recursive value function $L(\lambda)$ and the stacked Lagrangian value $\tilde{L}(\lambda)$ are rational in $\lambda$ and agree on the open region where the stacked Lagrangian argument applies; analytic continuation extends the equality to the interior $(\lambda_1, \infty)$ of the recursive domain, and the right limit at $\lambda = \lambda_1$ extends it to the boundary. (ii)~Differentiating $\tilde{L}(\lambda)$ at its saddle point, the partial derivatives in $\mathbf{u}$ and $\mathbf{w}$ are zero by stationarity, leaving $dL/d\lambda = \tfrac{1}{2} - \tfrac{1}{2\alpha}|\mathbf{z}^*(\lambda)|^2$. (iii)~Convexity of $L(\lambda)$ and the monotonicity of $dL/d\lambda$ give the optimality condition for $\lambda^*$. The full proof is in the appendix (\cref{app:proofs}).

\begin{remark}
We denote Lagrangian stationary disturbance variables by $z_k^*$, which are computed from the unconstrained stationary conditions \eqref{ndsc2}. The optimal disturbance $w_k^*$ satisfying the signal bound constraint is obtained via $w_k^* = \overline{w}_k(u^*_k) \cap \bbW$ as in \eqref{ndwcond2}.
\end{remark}

\subsection{Optimal Policy and Implementation}
\label{sec:policy-structure-signal}

The optimal state feedback policy for the SiDAR \eqref{signaldp-w} is nonlinear in the state. To understand this nonlinearity, we first recall the Bellman recursion from \eqref{bellman-recursion-signal}--\eqref{opt-control-policy-signal}
\begin{equation}
\begin{split}
V_k(x,b) = \min_u \max_{|w|^2 \leq b} \Big[ & \frac{1}{\alpha}\ell(x,u) \\
& + V_{k+1}(Ax+Bu+Gw, b - |w|^2) \Big]
\end{split}
\tag{\ref{bellman-recursion-signal}}
\end{equation}
\begin{equation}
\begin{split}
u_k^*(x,b) = \arg\min_u \max_{|w|^2 \leq b} \Big[ & \frac{1}{\alpha}\ell(x,u) \\
& + V_{k+1}(Ax+Bu+Gw, b - |w|^2) \Big]
\end{split}
\tag{\ref{opt-control-policy-signal}}
\end{equation}
The dynamic programming solution in \cref{prop:ndsignal} evaluates these recursions by introducing a single Lagrange multiplier $\lambda \geq 0$ for the aggregate budget constraint, transforming the problem into a backward recursion for the matrices $\Pi_k(\lambda)$ via \eqref{signalrec1} and a forward optimization for the multiplier at each stage. This approach eliminates the budget $b$ from the backward recursion: instead of representing $V_k(x,b)$ over the continuum $(x,b) \in \bbR^n \times [0,\alpha]$, the backward pass computes $\Pi_k(\lambda)$ as functions of $x$ alone, and the online optimization reduces to a scalar convex program over $\lambda$ at each measured state.

At stage $k$ with current state $x_k$ and remaining budget $b_k = \alpha - \sum_{j=0}^{k-1}|w_j|^2$, the optimal multiplier for the remaining $N-k$ stages is determined by
\begin{equation}
\lambda^*(x_k, b_k) = \arg\min_{\lambda \in [\lambda_{k+1}, \infty)} \frac{1}{2}\left(\frac{x_k}{\sqrt{\alpha}}\right)' \Pi_k(\lambda) \left(\frac{x_k}{\sqrt{\alpha}}\right) + \frac{b_k}{2 \alpha} \lambda
\label{eq:signal-k-opt}
\end{equation}
where $\lambda_{k+1}$ is the feasibility lower bound defined by
\begin{equation}
\begin{split}
\lambda_N &\eqbyd \norm{G'P_fG} \\
\lambda_{k+1} &\eqbyd
\begin{cases}
 \begin{aligned}
  &\min_{\lambda\ge \lambda_{k+2}}
      \Bigl\{
         \lambda : \lambda = \norm{ G'\Pi_{k+2}(\lambda)G}
      \Bigr\}
 \end{aligned} \\
 \qquad \qquad \text{if }\norm{G'\Pi_{k+2}(\lambda_{k+2})G}>\lambda_{k+2}\\[10pt]
 \lambda_{k+2} \\
 \qquad \qquad \text{if }\norm{G'\Pi_{k+2}(\lambda_{k+2})G}\le\lambda_{k+2}
\end{cases}
\end{split}
\label{eq:lambda-k-plus-1}
\end{equation}
ensuring existence of solutions to the Riccati recursion \eqref{signalrec1}, and $\Pi_k(\lambda)$ is computed via the backward recursion \eqref{signalrec1}. Given $\lambda^*(x_k, b_k)$, the optimal control from \eqref{ndsc2} is
\begin{equation}
u_k^*(x_k, \lambda^*(x_k,b_k)) = K_k(\lambda^*(x_k, b_k)) x_k
\label{eq:nonlinear-policy-signal}
\end{equation}
where the gain matrix is defined by
\begin{equation}
K_k(\lambda) \eqbyd -\begin{bmatrix} I & 0 \end{bmatrix} M_k(\lambda)^{-1} \begin{bmatrix} B' \\ G' \end{bmatrix} \Pi_{k+1}(\lambda) A
\label{eq:gain-def-signal}
\end{equation}
The policy \eqref{eq:nonlinear-policy-signal} is nonlinear in $x_k$ because the optimal multiplier $\lambda^*(x_k, b_k)$ depends on the state through the quadratic term in \eqref{eq:signal-k-opt}, making the composition $x_k \mapsto \lambda^*(x_k, b_k) \mapsto K_k(\lambda^*(x_k, b_k))$ state dependent and nonlinear.

\begin{remark}[Comparison with LQR]
Unlike standard LQR where backward dynamic programming computes fixed gain matrices $K_k$ that are applied directly as $u^*_k(x_k) = K_k x_k$, the SiDAR requires both a backward sweep (compute $\Pi_k(\lambda)$ for $\lambda \in [\lambda_{k+1}, \infty)$ via \eqref{signalrec1}) and an online forward optimization (solve \eqref{eq:signal-k-opt} at each stage $k$ given the current state $x_k$ and remaining budget $b_k$) to determine the state dependent gains. This online optimization introduces the nonlinearity.
\end{remark}

\begin{remark}[Implementation and time consistency]
The optimal policy \eqref{eq:nonlinear-policy-signal} requires resolving the optimization \eqref{eq:signal-k-opt} at each stage $k$ from the current state $x_k$ and remaining budget $b_k$. If the realized state deviates from the nominal trajectory, whether due to disturbances, model mismatch, or any other reason, the multiplier $\lambda^*(x_0)$ computed at $k=0$ is no longer optimal for the current state, and the optimization must be repeated. This shrinking horizon distinguishes the SiDAR from problems where the optimal policy can be precomputed offline.
\end{remark}

\begin{remark}[Computational implementation]
At each stage $k$, given the measured state $x_k$ and remaining budget $b_k$, the optimal policy \eqref{eq:nonlinear-policy-signal} is implemented online: solve the scalar optimization \eqref{eq:signal-k-opt} to obtain $\lambda^*(x_k, b_k)$, compute the gain $K_k(\lambda^*(x_k, b_k))$ from \eqref{eq:gain-def-signal}, and apply $u_k = K_k(\lambda^*(x_k, b_k)) x_k$. The optimization \eqref{eq:signal-k-opt} is a scalar convex program on $[\lambda_{k+1},\infty)$ (\cref{prop:ndsignal}) with derivative given in closed form by \cref{prop:derivative}: the optimum is at the boundary $\lambda^* = \lambda_{k+1}$ when $dL/d\lambda(\lambda_{k+1}) \geq 0$ (\cref{prop:derivative} item~3), and otherwise an interior root of $dL/d\lambda = 0$ is found by bisection or Brent's method.\footnote{In principle, the policy $u_k(x_k, b_k)$ could be precomputed offline as a function over a discretization of the state $x \in \bbR^n$ and budget $b \in [0, \alpha]$ and applied at runtime via table lookup; this suffers the curse of dimensionality and is limited to low dimensional systems. We do not pursue it here.} The online procedure is summarized in Algorithm~\ref{alg:online-signal}.
\end{remark}

\begin{remark}[Budget update]
\label{rem:sidar-budget-update}
The remaining budget is part of the augmented game state in the SiDAR. Given initial condition $b_0 = \alpha$, the budget evolves deterministically according to
\[
b_k = \alpha - \sum_{j=0}^{k-1} \normf{w_j}^2 \qquad b_{k+1} = b_k - \normf{w_k}^2
\]
and the controller uses $b_k$ directly when computing \eqref{eq:signal-k-opt}.
\end{remark}

\begin{algorithm}[h]
\caption{Online implementation of nonlinear optimal policy for SiDAR}
\label{alg:online-signal}
\begin{algorithmic}[1]
\STATE \textbf{Input:} Horizon $N$, system matrices $(A,B,G)$, weights $(Q,R,P_f)$, budget $\alpha$
\STATE Initialize remaining budget $b_0 = \alpha$
\FOR{$k = 0, 1, \ldots, N-1$}
    \STATE Measure current state $x_k$
    \STATE Solve optimization \eqref{eq:signal-k-opt} to obtain $\lambda^*(x_k, b_k)$
    \STATE Compute gain $K_k(\lambda^*(x_k, b_k))$ from \eqref{eq:gain-def-signal}
    \STATE Apply control $u_k = K_k(\lambda^*(x_k, b_k)) x_k$
    \STATE System evolves: $x_{k+1} = Ax_k + Bu_k + Gw_k$
    \STATE Update budget: $b_{k+1} = b_k - |w_k|^2$
\ENDFOR
\end{algorithmic}
\end{algorithm}

\section{Solution Regions and Properties}
\label{sec:control-region}
The SiDAR \eqref{signaldp-w} features two solution regions in the space of the initial state $x_0$ for a given disturbance budget $\alpha$. For notational simplicity, we develop the results for the initial problem with state $x_0$ and budget $\alpha$; the results apply at each stage $k$ with current state $x_k$ and remaining budget $b_k = \alpha - \sum_{j=0}^{k-1}|w_j|^2$ by replacing $x_0 \to x_k$, $\alpha \to b_k$, and using the tail problem from stage $k$ to $N$.

Let Assumptions 1-3 hold.
\begin{definition}[Solution regions for SiDAR] \hfill
\begin{enumerate}
    \item  Region $\mathcal{X}_{L}(\alpha) \subseteq \bbR^n$ is the initial states $x_0$ for which $\lambda^*(x_0) = \lambda_1$ is optimal in problem $\mathbf{L}_{si}$ \eqref{lsi}
    \item Region $\mathcal{X}_{NL}(\alpha) \subseteq \bbR^n$ is the initial states $x_0$ for which $\lambda^*(x_0) >  \lambda_1$ is optimal in problem $\mathbf{L}_{si}$ \eqref{lsi}
\end{enumerate}
\end{definition}

The solution region geometry is determined in \cref{prop:shape_signal}. Recall from \cref{prop:derivative} the disturbance stationary point $\mathbf{z}^*(\lambda) = \tilde{J}(\lambda) x_0$ where
\[
\tilde{J}(\lambda) \eqbyd \begin{bmatrix}
J_0(\lambda) \\
J_1(\lambda) \Phi_{1,0}(\lambda) \\
\vdots \\
J_{N-1}(\lambda) \Phi_{N-1,0}(\lambda)
\end{bmatrix}
\]
with $J_k(\lambda) \eqbyd -\begin{bmatrix} 0 & I \end{bmatrix} M_k(\lambda)^{-1} d_k(\lambda)$, $F_k(\lambda) \eqbyd A + BK_k(\lambda) + GJ_k(\lambda)$, and $\Phi_{k,j}(\lambda) \eqbyd F_{k-1}(\lambda)F_{k-2}(\lambda)\cdots F_j(\lambda)$ for $j < k$ and $\Phi_{k,k}(\lambda) \eqbyd I$.

\begin{proposition}[Region $\mathcal{X}_{L}(\alpha)$]
\label{prop:shape_signal}
The region $\mathcal{X}_{L}(\alpha)$ is given by
$$
\mathcal{X}_{L}(\alpha) = \left\{x_0 \in \bbR^n \,\bigg|\, \frac{x_0'}{\sqrt{\alpha}} \tilde{J}(\lambda_1)'\tilde{J}(\lambda_1)\frac{x_0}{\sqrt{\alpha}} \leq 1\right\}
$$
Thus, $\mathcal{X}_{L}(\alpha)$ is an ellipsoid centered at the origin.
\end{proposition}

\begin{proof}
From \cref{prop:derivative} item 3, the optimal multiplier satisfies $\lambda^* = \lambda_1$ if and only if $|\mathbf{z}^*(\lambda_1)|^2 \leq \alpha$. Since $\mathbf{z}^*(\lambda_1) = \tilde{J}(\lambda_1) x_0$, we have
\[
|\mathbf{z}^*(\lambda_1)|^2 = x_0' \tilde{J}(\lambda_1)'\tilde{J}(\lambda_1) x_0
\]
Therefore $\lambda^*(x_0) = \lambda_1$ if and only if $x_0' \tilde{J}(\lambda_1)'\tilde{J}(\lambda_1) x_0 \leq \alpha$, which is equivalent to
\[
\frac{x_0'}{\sqrt{\alpha}} \tilde{J}(\lambda_1)'\tilde{J}(\lambda_1)\frac{x_0}{\sqrt{\alpha}} \leq 1
\]
Since $\tilde{J}(\lambda_1)'\tilde{J}(\lambda_1) \succeq 0$, this defines an ellipsoid centered at the origin.
\end{proof}

\begin{corollary}[Region $\mathcal{X}_{NL}(\alpha)$]
\label{cor:shape_signal_nl}
Region $\mathcal{X}_{NL}(\alpha)$ is given by
$$\mathcal{X}_{NL}(\alpha) = \bbR^n \setminus \mathcal{X}_{L}(\alpha)$$
\end{corollary}

\begin{proof}
From \cref{prop:derivative} item 3, the optimal multiplier satisfies either $\lambda^* = \lambda_1$ or $\lambda^* > \lambda_1$. Since $\mathcal{X}_{L}(\alpha)$ defines all $x_0$ with $\lambda^*(x_0) = \lambda_1$ by \cref{prop:shape_signal}, the complement $\bbR^n \setminus \mathcal{X}_{L}(\alpha)$ defines all $x_0$ with $\lambda^*(x_0) > \lambda_1$.
\end{proof}

\Cref{prop:linear_signal} shows that within $\mathcal{X}_L(\alpha)$ the optimal control at the initial stage is linear in $x_0$ with a state independent gain $K_0(\lambda_1)$ obtained by evaluating the stationary conditions of \cref{prop:1dsignal} at $\lambda^* = \lambda_1$.

\begin{proposition}[Linear control in region $\mathcal{X}_{L}(\alpha)$]
\label{prop:linear_signal}
For a fixed $\alpha$ and $x_0 \in \mathcal{X}_L(\alpha)$, the optimal control policy at the initial stage is linear in the initial state
\begin{equation*}
    u^*_0(x_0) = K_0(\lambda_1) x_0
\end{equation*}
where the gain matrix $K_0(\lambda_1)$ is state independent and given by
\begin{align*}
K_0(\lambda_1) &= -\begin{bmatrix} I & 0 \end{bmatrix} 
      \begin{bmatrix}
         R + B'\Pi_{1}(\lambda_1)B & B'\Pi_{1}(\lambda_1)G \\
         G'\Pi_{1}(\lambda_1)B & G'\Pi_{1}(\lambda_1)G - \lambda_1 I
       \end{bmatrix}^{-1} \\
&\quad \begin{bmatrix}
         B'\Pi_{1}(\lambda_1)A \\[2pt] 
         G'\Pi_{1}(\lambda_1)A
       \end{bmatrix}
\end{align*}
and $\Pi_{1}(\lambda_1)$ is computed via the recursion \eqref{signalrec1}.
\end{proposition}

\begin{proof}
The optimal control $u^*_0(x_0)$ for \eqref{signaldp-w} is given by \eqref{ndsc2} at stage $k=0$
\[
 \begin{bmatrix} B'\Pi_{1}B + R & B'\Pi_{1} G \\ (B'\Pi_{1}G)' & G'\Pi_{1}G - \lambda^* I  \end{bmatrix}
\begin{bmatrix} u_0 \\ z_0 \end{bmatrix}^* = -\begin{bmatrix} B' \\ G'\end{bmatrix} \Pi_{1}A \; x_0
\]
For a fixed $\alpha$, if $x_0 \in \mathcal{X}_{L}(\alpha)$, then from the definition we have $\lambda^*(x_0) = \lambda_1$. Define $$M_0(\lambda^*) \eqbyd \begin{bmatrix} B'\Pi_{1}B + R & B'\Pi_{1} G \\ (B'\Pi_{1}G)' & G'\Pi_{1}G - \lambda^* I \end{bmatrix}$$ and $d_0 \eqbyd \begin{bmatrix} B'\Pi_{1}A \\ G'\Pi_{1}A \end{bmatrix}$. By the nonsingularity of $M_0(\lambda^*)$ we have
\[
\begin{bmatrix}u_0\\ z_0\end{bmatrix}^* = -M_0(\lambda^*)^{-1} d_0 x_0
\]
so $u_0^*(x_0)=-\begin{bmatrix}I&0\end{bmatrix}M_0(\lambda^*)^{-1}d_0 x_0 =: K_0(\lambda_1) x_0$. Since $\lambda^*(x_0) = \lambda_1$ is independent of $x_0$ for all $x_0 \in \mathcal{X}_L(\alpha)$, the gain matrix $K_0(\lambda_1)$ is constant (state independent), yielding a linear feedback policy in $x_0$. 
\end{proof}

\begin{remark}
    In the region \(\mathcal X_{NL}(\alpha)\) the solution
\(\lambda^*(x_0)\), and hence \(u^*(x_0)\), is a nonlinear
function of the initial state $x_0$ for a fixed $\alpha$.
\end{remark}

\begin{remark}
    For the region $\mathcal{X}_{NL}(\alpha)$, recursion \eqref{signalrec1} can be written as
    \[
    \Pi_{k} = Q + A' \Psi_{k+1}(I+B'R^{-1}B\Psi_{k+1})^{-1} A
    \]
    where $\Psi_{k+1} \eqbyd \Pi_{k+1} - \Pi_{k+1}G(G'\Pi_{k+1}G-\lambda I)^{-1}G'\Pi_{k+1} \succeq 0$.

Let Assumption 4 hold. Then recursion \eqref{signalrec1} applied in the region $\mathcal{X}_{NL}(\alpha)$ can be written as
    \[
    \Pi_{k} = Q + A' \Pi_{k+1}(I+(BR^{-1}B'-(1/\lambda)GG')\Pi_{k+1})^{-1}A
    \]
    This is an algebraic rewriting of \eqref{signalrec1}, valid for $\lambda > \norm{G'\Pi_{k+1} G}$, the condition that $\lambda^*(x_0) > \lambda_1$ satisfies throughout $\mathcal{X}_{NL}(\alpha)$. The same form was derived by \citet[p.\,86]{basar:bernhard:1995} for the standard $H_\infty$ state feedback problem at zero initial state. In our framework, zero initial state corresponds to $\mathcal{X}_L(\alpha)$ at $\lambda^* = \lambda_1$ (\cref{prop:linear_signal}), on the boundary of the recursive feasibility domain $[\lambda_1, \infty)$. The recursion is shared, but its scope here is broader: \citet{basar:bernhard:1995} cover one point of the state space (zero initial state, on the boundary), whereas the present paper uses the same recursion across the entire state space, with $\lambda^*(x_0) > \lambda_1$ in the interior of the feasibility domain when $x_0 \in \mathcal{X}_{NL}(\alpha)$. The closer prior work on nonzero initial states is \citet{didinsky:basar:1992}, who first identified a partition into regions analogous to $\mathcal{X}_L(\alpha)$ and $\mathcal{X}_{NL}(\alpha)$. They reformulated the problem into an auxiliary strongly dual one and did not give explicit solutions across both regions. The present derivation instead solves \eqref{signaldp-w} directly with the multiplier from \eqref{lsi}, in both regions.
\end{remark}

\begin{remark}
    The SiDAR \eqref{signaldp-w} features a unique solution region over the entire horizon length.
\end{remark}

\begin{remark}
    \label{regionzero}
    The ratio $\norm{x_0}/\sqrt{\alpha}$, which controls the size between the state and the disturbance, determines the region size. The zero state $x = 0$ is trivially contained in $\mathcal{X}_L(\alpha)$.
\end{remark}

\begin{remark}
The result in \cref{prop:shape_signal} applies at each stage $k$ with remaining budget $b_k = \alpha - \sum_{j=0}^{k-1}|w_j|^2$. At stage $k$ with current state $x_k$ and remaining budget $b_k$, the region where $\lambda^*(x_k, b_k) = \lambda_{k+1}$ is
\[
\mathcal{X}_L(b_k) = \left\{x_k \in \bbR^n \,\bigg|\, \frac{x_k'}{\sqrt{b_k}} \tilde{J}_k(\lambda_{k+1})'\tilde{J}_k(\lambda_{k+1})\frac{x_k}{\sqrt{b_k}} \leq 1\right\}
\]
where $\tilde{J}_k(\lambda_{k+1})$ is constructed from the remaining stages $k$ through $N-1$ and $\lambda_{k+1}$ is the feasibility lower bound from \cref{prop:ndsignal}.
While \cref{prop:linear_signal} establishes that the policy is linear in $x_0$ for fixed total budget $\alpha$, the optimal policy is nonlinear in $(x_k, b_k)$. At each stage, a new optimization over $\lambda$ must be solved with the remaining budget $b_k$, yielding an optimal multiplier $\lambda^*(x_k, b_k)$ and gain matrix $K_k(\lambda^*(x_k, b_k))$ that depend nonlinearly on $b_k$. Therefore, the policy $u_k = K_k(\lambda^*(x_k, b_k))x_k$ is nonlinear in $(x_k, b_k)$.
\end{remark}

\section{Numerical Examples}
\label{sec:num-example}
The following examples illustrate the finite horizon SiDAR. \Cref{sec:scalar-example} uses a scalar system to visualize the optimal policy and the partition into $\mathcal{X}_L$ and $\mathcal{X}_{NL}$. \Cref{sec:timing} sweeps the state dimension up to $n = 300$ to assess the execution time of \cref{alg:online-signal} for randomly generated systems with various $(m, q)$.

\subsection{Optimal control surface}
\label{sec:scalar-example}
Consider the scalar system
\begin{equation*}
A=0.5 \; \; \; B=1 \; \; \; G=1 \; \; \; R=1 \; \; \; Q=0.25 \; \; \; P_f=0.25
\end{equation*}
with horizon $N=10$ and disturbance budget $\alpha = 1$.

\Cref{fig:ustar_XL} illustrates the optimal control $u^*_0(x_0, \lambda^*(x_0,b_0))$ at the initial stage from \eqref{eq:nonlinear-policy-signal} as a function of state $x_0$ and remaining budget $b_0$. Note that $b_0 = \alpha$ at the initial stage. The top panel displays contour lines of $u^*_0(x_0, \lambda^*(x_0,b_0))$ with the shaded region indicating $\mathcal{X}_L(b_0)$ from \cref{prop:shape_signal}, where the optimal multiplier equals its lower bound $\lambda^* = \norm{G'\Pi_1(\lambda^*)G}$ and the policy is linear in $x_0$ for fixed $b_0$. The middle panel shows the cross-section $u^*(x_0, \lambda^*(x_0,1))$ for fixed budget $b_0=1$, with the shaded region indicating $\mathcal{X}_L(1)$. The bottom panel shows the cross-section $u^*_0(0.5, \lambda^*(0.5,b_0))$ for fixed state $x_0=0.5$, with the shaded region indicating the values of $b_0$ for which $(0.5, b_0) \in \mathcal{X}_L(b_0)$. Outside the shaded regions in all three panels, the policy is nonlinear as $\lambda^* > \norm{G'\Pi_1(\lambda^*)G}$, demonstrating the state-dependent transition between linear and nonlinear control regions characteristic of the SiDAR. At the boundary $b_0=0$, the policy recovers standard LQR control.

\begin{figure}
\centering
\includegraphics[width=1\linewidth]{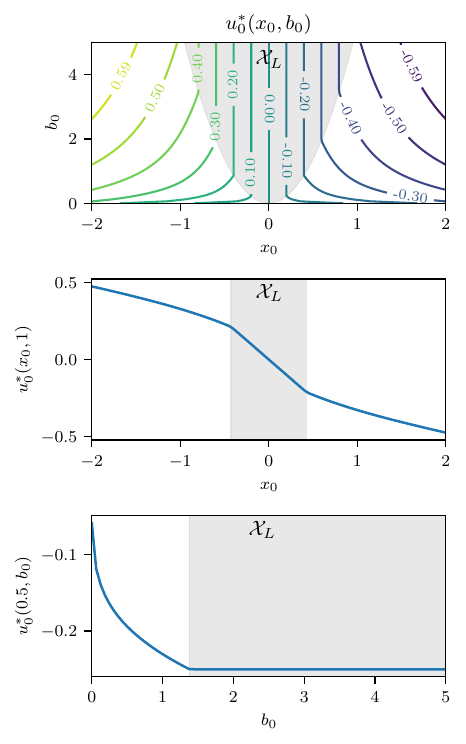}
\caption{Optimal control $u^*_0(x_0, b_0)$ for nondegenerate scalar system with $N=10$. Top: contour lines with shaded region $\mathcal{X}_L(b_0)$ where the policy is linear in $x_0$ for a fixed $b_0$. Middle: cross-section $u^*_0(x_0, 1)$ versus $x_0$ with shaded region $\mathcal{X}_L(1)$. Bottom: cross-section $u^*_0(0.5, b_0)$ versus $b_0$ with shaded region indicating $\mathcal{X}_L$ for fixed $x_0=0.5$. Note that $b_0 = \alpha$.}
\label{fig:ustar_XL}
\end{figure}

\subsection{Computational scaling}
\label{sec:timing}

\Cref{fig:timing} reports the execution time of \cref{alg:online-signal} at one stage as a function of the state dimension $n$, for randomly generated systems satisfying \cref{asst1,asst2,asst3,asst4}: $A$ is generated by sampling a standard Gaussian matrix and rescaling so that $\rho(A) = 1/1.05 \approx 0.95$; $B \in \bbR^{n \times m}$ has standard Gaussian entries; $G = BM$ with $M \in \bbR^{m \times q}$ Gaussian; $Q = I_n$, $R = I_m$, $P_f = 0.25\,I_n$, $N = 10$. Three combinations of $(m, q)$ are reported, each with five independent trials per dimension; the curves give the median execution time and the bands give the minimum and maximum across trials. The dominant cost per Riccati step is the dense linear algebra of the block solve and matrix multiplications, with overall complexity bounded by $O((n+m+q)^3)$; the curves approach this cubic reference at moderate to large $n$, while at small $n$ they saturate at the constant overhead of the implementation. The largest configuration $m = q = n = 300$ solves in approximately ten seconds.

\begin{figure}
\centering
\includegraphics[width=1\linewidth]{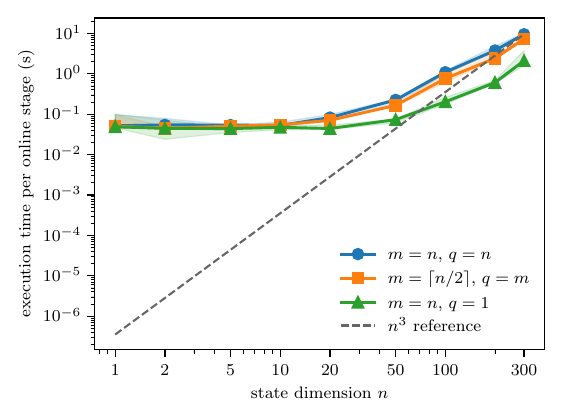}
\caption{Execution time of \cref{alg:online-signal} at one stage versus state dimension $n$, for three combinations of $(m, q)$ and randomly generated systems satisfying \cref{asst1,asst2,asst3,asst4}. Curves: median over five trials. Bands: minimum to maximum across trials. Dashed line: $cn^3$ reference, with $c$ chosen so the line passes through the $m = q = n$ curve at $n = 300$.}
\label{fig:timing}
\end{figure}

\section{Summary}
\label{sec:end}
This work presents a finite horizon recursive solution to the SiDAR for linear systems with arbitrary initial states. Existing theoretical results were limited to the zero initial state assumption, deriving policies valid only at the origin.

The optimal control policy at stage $k$ requires solving a tractable convex scalar optimization over the multiplier $\lambda$ given the current state $x_k$ and remaining disturbance budget $b_k$; the control gain is then explicit. The backward Riccati recursion operates in the state $x$ alone. The resulting control law is nonlinear in $x$ through the dependence of the optimal multiplier $\lambda^*(x_k, b_k)$ on the state. For fixed remaining budget $b_k$ at stage $k$, the state space partitions into two distinct regions: $\mathcal{X}_L(b_k)$, where the control policy is linear in $x$, and $\mathcal{X}_{NL}(b_k)$, where the control policy is nonlinear in $x$.

We establish monotonicity and boundedness of the associated Riccati recursion. The region $\mathcal{X}_L(b_k)$ is an ellipsoid centered at the origin, whose geometry is determined by the Lagrangian stationary disturbance gains. The derivative of the value function with respect to the Lagrange multiplier provides optimality conditions that distinguish the linear and nonlinear solution regions.

A companion paper \citep{mannini:rawlings:2026b} extends these results to the infinite horizon setting, classifying systems as degenerate or nondegenerate, establishing convergence properties, and reducing the infinite horizon problem to a tractable LMI optimization.

\section{Appendix}
\label{sec:props}

In this appendix, we compile the fundamental results used throughout this paper.

The following classical result justifies the interchange of minimization and maximization in the Lagrangian analysis of \Cref{sec:signal}.

\begin{theorem}[Minimax Theorem]
\label{th:minimax}
Let $U \subset \bbR^m$ and $W \subset \bbR^q$ be compact convex sets. If 
$V: U \times W \to \bbR$ is a continuous function that is convex-concave, i.e.,
$V(\cdot ,w):U \to \bbR$ is convex for all $w \in W$, and
$V(u, \cdot ):W \to \bbR$ is concave for all $u \in U$\\
Then we have that
\begin{equation*}
\min_{u \in U} \max_{w \in W} V(u,w) = \max_{w \in W} \min_{u \in U} V (u,w) 
\end{equation*}
\end{theorem}

The following result guarantees invertibility of the block matrix $M_k(\lambda)$ appearing in the Riccati recursion \eqref{signalrec1}, ensuring well-posedness of the finite horizon solution under Assumption~\ref{asst2}.

\begin{proposition}[Invertibility under range inclusion]
\label{prop:range-inclusion-invertible}
Let ${\Pi}\succeq 0$, $R\succ 0$, and ${\lambda}>0$. Assume $G'{\Pi}G-{\lambda}I\preceq 0$. If $\mathcal{R}(G)\subseteq \mathcal{R}(B)$ (equivalently, $\mathcal{N}(B')\subseteq \mathcal{N}(G')$), then the block matrix
\[
   M \eqbyd
   \begin{bmatrix}
      B'{\Pi}B + R & B'{\Pi}G \\
      G' \Pi B & G'{\Pi}G - {\lambda} I
   \end{bmatrix}
\]
is nonsingular.
\end{proposition}

\begin{proof}
Write $A\eqbyd B'{\Pi}B+R\succ 0$, $C\eqbyd B'{\Pi}G$, and $D\eqbyd G'{\Pi}G-{\lambda}I\preceq 0$, so that
$M=\begin{bmatrix} A & C \\ C' & D \end{bmatrix}$. Using the partitioned matrix determinant formula, since $A\succ 0$, the matrix $M$ is invertible if and only if its Schur complement $\tM \eqbyd D - C' A^{-1} C$ is nonsingular.

To establish sufficiency, assume $\mathcal{R}(G)\subseteq \mathcal{R}(B)$ and let $v$ satisfy $\tM v = 0$. Then
\[
v'\tM v \;=\; 0 \;=\; v'Dv - (Cv)'A^{-1}(Cv).
\]
Since $D\preceq 0$ and $A^{-1}\succ 0$, we have $v'Dv\le 0$ and $(Cv)'A^{-1}(Cv)\ge 0$, and because they are equal, both are zero:
$v'Dv=(Cv)'A^{-1}(Cv)=0$. As $-D\succeq 0$ and $A^{-1}\succ 0$ admit square roots, this implies $\sqrt{-D}\,v=0$ and $A^{-1/2}Cv=0$, hence $Dv=0$ and $Cv=0$.

From $Cv=B'{\Pi}Gv=0$ we obtain ${\Pi}Gv\in\mathcal{N}(B')\subseteq\mathcal{N}(G')$, so $G'\Pi Gv=0$. Together with $Dv=(G'{\Pi}G-{\lambda}I)v=0$ and ${\lambda}>0$, it follows that $v=0$. Therefore $\tM$ is nonsingular, and hence $M$ is nonsingular.
\end{proof}

The following lemma establishes equivalent forms of the Riccati recursion, expressing the value matrix $\Pi(\lambda)$ in terms of closed-loop quantities. This representation is used in the monotonicity analysis of \Cref{prop:ndsignal}.

\begin{lemma}[Riccati equalities]
\label{receq}
The equality 
\begin{align*}
\Pi(\lambda) &= Q+A'\Pi A- A' \Pi \begin{bmatrix} B & G \end{bmatrix} \\
&\quad \begin{bmatrix} B'\Pi B + R & B'\Pi G \\ (B'\Pi G)' & G'\Pi G - \lambda I  \end{bmatrix}^{\dagger}
    \begin{bmatrix} B' \\ G'\end{bmatrix}\Pi A
\end{align*}
    can be rewritten as
\begin{equation}
            \Pi = \bar{Q} + \bar{A}' \Pi\bar{A} - \bar{A}'\Pi G(G'\Pi G-\lambda I)^{\dagger}G'\Pi\bar{A} \label{xlrec}
        \end{equation}
    where $\bar{A} = A+BK$ and $\bar{Q} = Q + K'RK$ and $K$ satisfies
    \begin{equation}
  \begin{bmatrix}
        B'\Pi B+R & B'\Pi G \\
        G'\Pi B & G'\Pi G-\lambda I
    \end{bmatrix}\begin{bmatrix}
K \\ J
\end{bmatrix} =-\begin{bmatrix}
B'\Pi A \\ G'\Pi A
\end{bmatrix} \label{JK}
\end{equation}
\end{lemma}

\begin{proof}
From $M^\dagger M M^\dagger=M^ \dagger$ we have
 \begin{align}
         \Pi(\lambda) &=Q+A'\Pi A- A' \Pi \begin{bmatrix} B & G \end{bmatrix}
     M(\lambda)^{\dagger}
    \begin{bmatrix} B' \\ G'\end{bmatrix}\Pi A  \label{recjk} \\ &= Q+A'\Pi A- A' \Pi \begin{bmatrix} B & G \end{bmatrix}
     M(\lambda)^{\dagger}M(\lambda)M(\lambda)^{\dagger}
    \begin{bmatrix} B' \\ G'\end{bmatrix}\Pi A \nonumber
    \end{align}
where
\[
M(\lambda) = \begin{bmatrix} B'\Pi B + R & B'\Pi G \\ (B'\Pi G)' & G'\Pi G - \lambda I  \end{bmatrix}
\]
Define
\[
\begin{bmatrix}
        B'\Pi B+R & B'\Pi G \\
        G'\Pi B & G'\Pi G-\lambda I
    \end{bmatrix}\begin{bmatrix}
K \\ J
\end{bmatrix} =-\begin{bmatrix}
B'\Pi A \\ G'\Pi A
\end{bmatrix}
\]
or equivalently, with $b = \begin{bmatrix} B'\Pi A \\ G'\Pi A \end{bmatrix}$
    \begin{equation}
        \begin{bmatrix}
K \\ J
\end{bmatrix} =-M(\lambda)^{\dagger}b+\mathcal{N}(M(\lambda)) \label{JK2}
    \end{equation}
For any $v \in \mathcal{N}(M(\lambda))$, we have $M(\lambda)v = 0$, which gives $v'M(\lambda)v = 0$ and $(M(\lambda)^{\dagger}b)'M(\lambda)v = b'M(\lambda)^{\dagger}M(\lambda)v = 0$. Therefore, when substituting $\begin{bmatrix} K \\ J \end{bmatrix} = -M(\lambda)^{\dagger}b + v$ into the quadratic form $\begin{bmatrix} K' & J' \end{bmatrix} M(\lambda) \begin{bmatrix} K \\ J\end{bmatrix}$, all terms involving $v$ vanish. Thus, the following expression
\begin{equation}
        \Pi(\lambda) =Q+A'\Pi A- \begin{bmatrix} K' & J' \end{bmatrix}
     M(\lambda)
    \begin{bmatrix} K \\ J\end{bmatrix} \label{kj2}
\end{equation}
is equivalent to \eqref{recjk}. Expanding \eqref{kj2}
\begin{align*}
    \Pi(\lambda) = & Q+A'\Pi A-K'B'\Pi BK -K'RK-K'B'\Pi GJ-J'G'\Pi BK\\&-J'(G'\Pi G-\lambda I)J
\end{align*}
Consider
\[
B'\Pi GJ = -B' \Pi A - (B'\Pi B + R)K
\]
and
\[
J = -(G' \Pi G-\lambda I)^{\dagger}G' \Pi (A+BK)+\mathcal{N}(G' \Pi G-\lambda I)
\]
For any $q \in \mathcal{N}(G' \Pi G-\lambda I)$, we have $(G' \Pi G-\lambda I)q = 0$, giving $q'(G' \Pi G-\lambda I)q = 0$ and $((G' \Pi G-\lambda I)^{\dagger}G' \Pi (A+BK))'(G' \Pi G-\lambda I)q = 0$. Therefore, all terms involving $q$ vanish in the quadratic form $J'(G'\Pi G-\lambda I)J$.
Thus, substituting $B'\Pi GJ$ and $J$ in $\Pi(\lambda) = Q+A'\Pi A-K'B'\Pi BK -K'RK-K'B'\Pi GJ-J'G'\Pi BK-J'(G'\Pi G-\lambda I)J$ we obtain
\begin{align*}
\Pi(\lambda) =& Q+K'RK+(A+BK)'\Pi(A+BK)\\&-(A+BK)'\Pi G(G'\Pi G-\lambda I)^{\dagger}G' \Pi(A+BK)
\end{align*}
which is \eqref{xlrec} with $\bar{A} = A+BK$ and $\bar{Q} = Q + K'RK$.
\end{proof}

\subsection{\textbf{Proofs of Main Results}}
\label{app:proofs}

\begin{proof}[Proof of \cref{prop:1dsignal}]
Define
\begin{align*}
   &M_0(\lambda) \eqbyd \begin{bmatrix}
        B'\Pi_1B+R & B'\Pi_1G \\
        G'\Pi_1B & G'\Pi_1G-\lambda I
    \end{bmatrix} \qquad d_0 \eqbyd \begin{bmatrix}
B'\Pi_1A \\ G'\Pi_1A
\end{bmatrix}x_0  \\
       &M_1(\lambda) \eqbyd \begin{bmatrix}
        B'P_{f}B+R & B'P_{f}G \\
        G'P_{f}B & G'P_{f}G-\lambda I
    \end{bmatrix} \qquad d_1 \eqbyd \begin{bmatrix}
B'P_{f}A \\ G'P_{f}A
\end{bmatrix}x_1
   \end{align*}

\textbf{Stacked problem and Lagrangian setup.} \newline
Represent the linear system \eqref{system} in stacked form as
\begin{gather*}
\mathbf{x} =  \mathcal{A} x_0 + \mathcal{B} \mathbf{u} + \mathcal{G} \mathbf{w}
\end{gather*}
where $\mathbf{x} \eqbyd (x_1, x_2)$, $\mathbf{u} \eqbyd (u_0, u_1)$, $\mathbf{w} \eqbyd (w_0, w_1)$, and
\begin{gather*}
\mathcal{A} \eqbyd \begin{bmatrix} A \\ A^2 \end{bmatrix} \quad
\mathcal{B} \eqbyd  \begin{bmatrix} B & 0 \\ AB & B \end{bmatrix} \quad
\mathcal{G} \eqbyd \begin{bmatrix} G & 0 \\ AG & G \end{bmatrix}
\end{gather*}
Define the block diagonal weight matrices
\begin{equation*}
\mathcal{Q} \eqbyd \mathrm{diag}(Q, P_f) \qquad \mathcal{R} \eqbyd \mathrm{diag}(R, R)
\end{equation*}
The objective function is
\begin{align*}
V(x_0, \mathbf{u}, \mathbf{w}) &= (1/2)(x_0'Qx_0+\mathbf{x}'\mathcal{Q}\mathbf{x}+\mathbf{u}'\mathcal{R}\mathbf{u})
\end{align*}

Consider the stacked optimization $$\min_{\mathbf{u}}\max_{|\mathbf{w}|^2 \leq \alpha} \frac{V(x_0,\mathbf{u},\mathbf{w})}{\sum^{1}_{k=0} |w_k|^2}$$
We first show that the inequality constraint can be replaced by equality. Substituting the state dynamics into the cost yields
\[
V(x_0, \mathbf{u}, \mathbf{w}) = \frac{1}{2}\mathbf{w}'\mathcal{G}'\mathcal{Q}\mathcal{G}\mathbf{w} + \mathbf{w}'\mathcal{G}'\mathcal{Q}(\mathcal{A}x_0 + \mathcal{B}\mathbf{u}) + c(\mathbf{u})
\]
where $c(\mathbf{u})$ is independent of $\mathbf{w}$. For fixed $(x_0, \mathbf{u})$, the maximization over $\mathbf{w}$ is a convex quadratic plus a linear term. Since $\mathcal{Q} \succeq 0$, we have $\mathcal{G}'\mathcal{Q}\mathcal{G} \succeq 0$. Under \cref{asst3}, $G'P_fG \neq 0$, which through the definitions of $\mathcal{G}$ and $\mathcal{Q}$ ensures $\mathcal{G}'\mathcal{Q}\mathcal{G} \neq 0$. Assume for contradiction that an unconstrained maximum over $\mathbf{w}$ exists. This requires $\mathcal{G}'\mathcal{Q}\mathcal{G} \prec 0$. However, since $\mathcal{Q} \succeq 0$, we have $\mathcal{G}'\mathcal{Q}\mathcal{G} \succeq 0$, which implies $\mathcal{G}'\mathcal{Q}\mathcal{G} = 0$, contradicting \cref{asst3}. Therefore, the maximum over the constraint $|\mathbf{w}|^2 \leq \alpha$ occurs on the boundary $|\mathbf{w}|^2 = \alpha$, and we can equivalently consider
\[
\min_{\mathbf{u}}\max_{|\mathbf{w}|^2 = \alpha}\frac{V(x_0,\mathbf{u},\mathbf{w})}{\alpha}
\]

This problem has a different information structure from the sequential optimization \eqref{1dsignaldp}: in the stacked problem all components of $\mathbf{u}$ have full knowledge of all components of $\mathbf{w}$, whereas in the sequential problem each $u_k$ knows only $w_0,\ldots,w_{k-1}$ but not $w_k,\ldots,w_{N-1}$. Define the Lagrangian function
\[
L(x_0, \mathbf{u}, \mathbf{w}, \lambda) \eqbyd  V(x_0, \mathbf{u}, \mathbf{w}) - (\lambda/2) \left( \mathbf{w}'\mathbf{w} - \alpha \right)
\]
By applying Proposition 7 from Rawlings et al.\citep{rawlings:mannini:kuntz:2024,rawlings:mannini:kuntz:2024b} to the equality constrained stacked optimization we obtain
\[
\min_{\mathbf{u}}\max_{|\mathbf{w}|^2 = \alpha} \frac{V(x_0, \mathbf{u}, \mathbf{w})}{\alpha} = (1/\alpha)\min_{\mathbf{u}}\max_{\mathbf{w}} \min_{\lambda} L(x_0, \mathbf{u}, \mathbf{w}, \lambda)
\]
Using the stacked system representation, the Lagrangian becomes
\begin{align*}
L(x_0, \mathbf{u}, \mathbf{w}, \lambda) &= (1/2)(x_0'Qx_0+\mathbf{x}'\mathcal{Q}\mathbf{x}+\mathbf{u}'\mathcal{R}\mathbf{u}-\lambda\mathbf{w}'\mathbf{w}+\lambda\alpha)\\
&= (1/2) \begin{bmatrix} \mathbf{u} \\ \mathbf{w} \end{bmatrix}'
\begin{bmatrix} \mathcal{B}'\mathcal{Q}\mathcal{B}+\mathcal{R} & \mathcal{B}'\mathcal{Q}\mathcal{G} \\ (\mathcal{B}'\mathcal{Q}\mathcal{G})' & \mathcal{G}'\mathcal{Q}\mathcal{G}-\lambda I \end{bmatrix}
\begin{bmatrix} \mathbf{u} \\ \mathbf{w} \end{bmatrix} \\
&\quad + \begin{bmatrix} \mathbf{u} \\ \mathbf{w} \end{bmatrix}' \begin{bmatrix} \mathcal{B}'\mathcal{Q}\mathcal{A} \\ \mathcal{G}'\mathcal{Q}\mathcal{A} \end{bmatrix} x_0 \\
&\quad + (1/2)x_0'(Q+\mathcal{A}'\mathcal{Q}\mathcal{A})x_0 + \lambda\alpha/2
\end{align*}

For $\lambda \geq \norm{\mathcal{G}'\mathcal{Q}\mathcal{G}}$, we have $\mathcal{G}'\mathcal{Q}\mathcal{G}-\lambda I \preceq 0$, hence $L(x_0,\mathbf{u},\mathbf{w},\lambda)$ is concave in $\mathbf{w}$ for fixed $(x_0,\mathbf{u},\lambda)$. By Proposition 15 (strong duality for sphere constrained quadratic in $\mathbf{w}$) from Rawlings et al.\citep{rawlings:mannini:kuntz:2024,rawlings:mannini:kuntz:2024b}, for every fixed $(x_0, \mathbf{u})$, we have
\[
\max_{|\mathbf{w}|^2 = \alpha} V(x_0, \mathbf{u}, \mathbf{w}) = \min_{\lambda \geq \norm{\mathcal{G}'\mathcal{Q}\mathcal{G}}} \max_{\mathbf{w}} L(x_0, \mathbf{u}, \mathbf{w}, \lambda)
\]
Hence, we obtain
\begin{align*}
\min_{\mathbf{u}} \max_{|\mathbf{w}|^2 = \alpha} V(x_0, \mathbf{u}, \mathbf{w}) 
&= \min_{\mathbf{u}} \min_{\lambda \geq \norm{\mathcal{G}'\mathcal{Q}\mathcal{G}}} \max_{\mathbf{w}} L \\
&= \min_{\lambda \geq \norm{\mathcal{G}'\mathcal{Q}\mathcal{G}}} \min_{\mathbf{u}} \max_{\mathbf{w}} L
\end{align*}
where the last equality follows from interchanging the order of minimization.

From Proposition 12.a in Rawlings et al. \citep{rawlings:mannini:kuntz:2024,rawlings:mannini:kuntz:2024b}, since $\mathcal{B}'\mathcal{Q}\mathcal{B}+\mathcal{R} \succ 0$ (from $\mathcal{Q}\succeq 0$ and $\mathcal{R}\succ 0$) and $\mathcal{G}'\mathcal{Q}\mathcal{G}-\lambda I \preceq 0$ for $\lambda \geq \norm{\mathcal{G}'\mathcal{Q}\mathcal{G}}$, strong duality holds between the minimization over $\mathbf{u}$ and maximization over $\mathbf{w}$ in the stacked Lagrangian. Therefore $\min_{\mathbf{u}}\max_{\mathbf{w}} L = \max_{\mathbf{w}}\min_{\mathbf{u}} L$, and more generally, all orderings of the individual $\min_{u_k}$ and $\max_{w_k}$ operations yield the same value. In particular, for any $\lambda \geq \norm{\mathcal{G}'\mathcal{Q}\mathcal{G}}$, we have
\[
\min_{\mathbf{u}}\max_{\mathbf{w}} L = \min_{u_0}\max_{w_0} \min_{u_1}\max_{w_1} L
\]
Combining with the interchange of $\min_{\lambda}$ established above, we obtain
\begin{align*}
V^*(x_0) = (1/\alpha)\min_{\lambda \geq \norm{\mathcal{G}'\mathcal{Q}\mathcal{G}}} \min_{u_0}\max_{w_0}  &\min_{u_1}\max_{w_1} \\
&\quad L(x_0, u_0,w_0,u_1,w_1, \lambda)
\end{align*}
where
\begin{align*}
L(x_0, u_0,w_0,u_1,w_1, \lambda) &= \ell(x_0,u_0)+\ell(x_1,u_1) + \ell_f(x_2) \\
&\quad - (\lambda/2)(w'_0w_0+w'_1w_1-\alpha)
\end{align*}
and the minimization over $\lambda$ is in the outermost position.

The bound $\lambda \geq \norm{\mathcal{G}'\mathcal{Q}\mathcal{G}}$ from the stacked problem establishes the existence of a sufficiently large $\lambda$ for which strong duality holds, guaranteeing that all orderings of the individual $\min_{u_k}$ and $\max_{w_k}$ operations yield the same value. This existence result justifies placing $\min_{\lambda}$ in the outermost position. Having established this, we now solve the sequential dynamic programming, which exploits the causal information structure: at each stage $k$, the control $u_k$ is chosen with knowledge of only $w_0, \ldots, w_{k-1}$, not the future disturbances $w_k, \ldots, w_{N-1}$. This nested optimization admits stagewise feasibility conditions that are propagated backward to determine the feasibility bound $\lambda_1$, defining the recursive feasibility domain $[\lambda_1, \infty)$ for which the sequential minmax problem admits solutions at every stage.

\textbf{First step: from $k=2$ to $k=1$.} \newline
Since $\ell(x_0,u_0)$ is independent of $(u_1,w_1)$ once $(u_0,w_0)$ are fixed, we can rewrite
\begin{align*}
V^*(x_0) &= (1/ \alpha)\min_{\lambda}\min_{u_0}\max_{w_0}\Big[\ell(x_0,u_0) -(\lambda/2)(w'_0w_0-\alpha)\\
&\quad +  \min_{u_1}\max_{w_1} \big(\ell(x_1,u_1) + \ell_f(x_2) - (\lambda/2)(w'_1w_1)\big) \Big]
\end{align*}
The term $\ell(x_1,u_1) + \ell_f(x_2) - (\lambda/2)(w'_1w_1)$ is equivalent to
\[
(1/2)\;  \begin{bmatrix} u_1 \\ w_1 \end{bmatrix}'
M_1(\lambda) \begin{bmatrix} u_1 \\ w_1 \end{bmatrix}  +
 \begin{bmatrix} u_1 \\ w_1 \end{bmatrix}' d_1 + (1/2) \; x_1'(Q+A'P_fA)x_1 
\]
Define $\lambda_2\eqbyd\norm{G'P_fG} \neq 0$. Applying Proposition 14.a from Rawlings et al.\citep{rawlings:mannini:kuntz:2024,rawlings:mannini:kuntz:2024b} to $\min_{u_1}\max_{w_1} [\ell(x_1,u_1) + \ell_f(x_2) - (\lambda/2)(w'_1w_1)]$ yields
\begin{align*}
(1/ \alpha)\min_{u_1}\max_{w_1} &[\ell(x_1,u_1) + \ell_f(x_2) - (\lambda/2)(w'_1w_1)] \\
&= (1/2\alpha)x_1' \Pi_1(\lambda)x_1
\end{align*}
where
\begin{align*}
\Pi_1(\lambda) &=Q+A'P_fA \\
&\quad -A' P_f \begin{bmatrix} B & G \end{bmatrix}
     \begin{bmatrix} B'P_fB + R & B'P_f G \\ (B'P_fG)' & G'P_fG - \lambda I  \end{bmatrix}^{\dagger}
    \begin{bmatrix} B' \\ G'\end{bmatrix} P_fA
\end{align*}
which from \cref{receq} can be rewritten as
$$\Pi_{1}(\lambda) = \bar{Q}_{1} + \bar{A}_{1}'P_f\bar{A}_{1} - \bar{A}_{1} P_f G(G'P_f G-\lambda I)^{\dagger}G'P_f \bar{A}_{1}$$
where $\bar{A}_{1} = A+BK_{1}$ and $\bar{Q}_{1} = Q + K_{1}'RK_{1}$ and $K_{1}$ satisfies
    \begin{equation*}
  \begin{bmatrix}
        B'P_f B+R & B'P_f G \\
        G'P_f B & G'P_f G-\lambda I
    \end{bmatrix}\begin{bmatrix}
K_{1} \\ J_{1}
\end{bmatrix} =\begin{bmatrix}
B'P_f A \\ G'P_f A
\end{bmatrix} 
\end{equation*} 
From $Q\succeq0$ and $R\succ0$, we have $\bar{Q}_{1}\succeq0$. From $G'P_f G-\lambda I \preceq0$, we have $\bar{A}_{1} P_f G(G'P_f G-\lambda I)^{\dagger}G'P_f \bar{A}_{1} \preceq 0$. From $\bar{Q}_{1}\succeq0$, $G'P_f G-\lambda I \preceq0$, $\bar{A}_{1} P_f G(G'P_f G-\lambda I)^{\dagger}G'P_f \bar{A}_{1} \preceq 0$, and $P_f \succeq 0$, we have $\Pi_{1}(\lambda) \succeq0$ for $\lambda \geq \lambda_2$.

From Proposition 14.a from Rawlings et al.\citep{rawlings:mannini:kuntz:2024,rawlings:mannini:kuntz:2024b}, solutions to $\min_{u_1}\max_{w_1} [\ell(x_1,u_1) + \ell_f(x_2) - (\lambda/2)(w'_1w_1)]$ exist for $\lambda = \lambda_2$ for $d_1 \in \mathcal{R}(M_1(\lambda_2))$ and for $\lambda > \lambda_2$ for all $d_1 \in \bbR^{m+q}$. From \cref{prop:range-inclusion-invertible}, $G'P_f G-\lambda I \preceq0$, and Assumptions 2-3, $M_1(\lambda)$ is invertible for $\lambda \geq\lambda_2$. Thus, $d_1 \in \mathcal{R}(M_1(\lambda_2))$ is always satisfied, and solutions exist for $\lambda \geq \lambda_2$ for all $d_1\in \bbR^{m+q}$. 

Define $\phi_1(\lambda,x_1)\eqbyd \min_{u_1}\max_{w_1}[\ell(x_1,u_1) + \ell_f(x_2) - (\lambda/2)(w'_1w_1)]$. We establish joint convexity in $(\lambda,x_1)$. For fixed $w_1$, the function $(u_1,\lambda,x_1) \mapsto \ell(x_1,u_1) + \ell_f(Ax_1+Bu_1+Gw_1) - (\lambda/2)w_1'w_1$ is convex by composition of convex functions with affine mappings \citep[§3.2.4]{boyd:vandenberghe:2004}. For fixed $(u_1,\lambda,x_1)$, the map $w_1 \mapsto \ell(x_1,u_1) + \ell_f(Ax_1+Bu_1+Gw_1) - (\lambda/2)w_1'w_1$ is concave for $\lambda \geq \lambda_2$ since $G'P_fG - \lambda I \preceq 0$. Therefore $g_1(u_1,\lambda,x_1)\eqbyd\max_{w_1}[\ell(x_1,u_1) + \ell_f(Ax_1+Bu_1+Gw_1) - (\lambda/2)w_1'w_1]$ is the pointwise supremum of convex functions in $(u_1,\lambda,x_1)$, hence convex by \citep[§3.2.3]{boyd:vandenberghe:2004} (see also \citet[Theorem 5.5]{rockafellar:1970}). The partial minimization $\phi_1(\lambda,x_1) = \min_{u_1} g_1(u_1,\lambda,x_1)$ preserves joint convexity in $(\lambda,x_1)$ by \citep[§3.2.5]{boyd:vandenberghe:2004} (see also \citet[Theorem 5.3]{rockafellar:1970}).

\textbf{Second step: from $k=1$ to $k=0$.} \newline
Proceeding to the next stage, we have
\begin{align*}
V^*(x_0) &= (1/ \alpha)\min_{\lambda} \bigg[ \lambda\alpha/2 + \min_{u_0}\max_{w_0}  [\ell(x_0,u_0)+ x_1' \Pi_1(\lambda)x_1 \\
&\quad - (\lambda/2)(w'_0w_0)]  \bigg]
\end{align*}
The term $\ell(x_0,u_0)+ x_1' \Pi_1(\lambda)x_1 - (\lambda/2)(w'_0w_0)$ is equivalent to
\begin{align*}
&(1/2)\;  \begin{bmatrix} u_0 \\ w_0 \end{bmatrix}'
M_0(\lambda) \begin{bmatrix} u_0 \\ w_0 \end{bmatrix}  +
 \begin{bmatrix} u_0 \\ w_0 \end{bmatrix}' d_0 \\
&\quad + (1/2) \; x_0'(Q+A'\Pi_1(\lambda)A)x_0
\end{align*}
Applying Proposition 14.a from Rawlings et al.\citep{rawlings:mannini:kuntz:2024,rawlings:mannini:kuntz:2024b} to $\min_{u_0}\max_{w_0}  [\ell(x_0,u_0)+ x_1' \Pi_1(\lambda)x_1 - (\lambda/2)(w'_0w_0)]$ for $\lambda\geq\lambda_2$ yields
\begin{align*}
(1/ \alpha)\min_{u_0}\max_{w_0}  &[\ell(x_0,u_0)+ x_1' \Pi_1(\lambda)x_1 \\
&\quad - (\lambda/2)(w'_0w_0)] = (1/2\alpha)x_0'\Pi_0(\lambda)x_0
\end{align*}
 where
\begin{align*}
\Pi_0(\lambda) &= Q+A'\Pi_1(\lambda)A -A' \Pi_1(\lambda) \begin{bmatrix} B & G \end{bmatrix} \\
&\quad \begin{bmatrix} B'\Pi_1(\lambda)B + R & B'\Pi_1(\lambda) G \\ (B'\Pi_1(\lambda)G)' & G'\Pi_1(\lambda)G - \lambda I  \end{bmatrix}^{\dagger}
    \begin{bmatrix} B' \\ G'\end{bmatrix} \Pi_1(\lambda)A
\end{align*}
We analyze for which conditions the solutions to $$\min_{u_0}\max_{w_0}  [\ell(x_0,u_0)+ x_1' \Pi_1(\lambda)x_1 - (\lambda/2)(w'_0w_0)]$$ exist. By doing so, we defer the Lagrange multiplier $\lambda$, an optimization variable, from the first step to the second step, and eventually to an outer scalar optimization.

See \cref{fig:defer} for visualizing the following argument. From Proposition 14.a from Rawlings et al.\citep{rawlings:mannini:kuntz:2024,rawlings:mannini:kuntz:2024b}, solutions to $\min_{u_0}\max_{w_0}  [\ell(x_0,u_0)+ x_1' \Pi_1(\lambda)x_1 - (\lambda/2)(w'_0w_0)]$ exist if $\lambda \geq \norm{G'\Pi_1(\lambda)G}$. We guarantee that the inequality $\lambda \geq \norm{G'\Pi_1(\lambda)G}$ holds by constructing $\lambda_1\ge\lambda_2$ such that, for all $\lambda\ge\lambda_1$, the admissibility condition $\lambda \geq \norm{G'\Pi_1(\lambda)G}$ holds. Hence the inner minmax problem admits a saddle point for every $\lambda\ge\lambda_1$, and the choice of $\lambda$ can be deferred to the outer (stage $k = 0$) scalar optimization.

For $\lambda\geq\lambda_2$ define
\[
m(\lambda)\eqbyd\norm{G'\Pi_1(\lambda)G}
\]
Since $M_1(\lambda)$ is invertible on $[\lambda_2,\infty)$ and all operators used to build $\Pi_1(\lambda)$ e.g., inverse of a matrix, are continuous there, $\Pi_1(\lambda)$ and $m(\lambda)$ are continuous on $[\lambda_2,\infty)$.
To prove that $m(\lambda)$ is nonincreasing, fix $\lambda_+\ge\lambda_-\ge\lambda_2$.  
Define 
\begin{align*}
q(u,w) &\eqbyd x'Qx+u'Ru+(Ax\!+\!Bu\!+\!Gw)'P_f(Ax\!+\!Bu\!+\!Gw)\\
&\quad -(\lambda_+/2)\,w'w\\
r(u,w) &\eqbyd x'Qx+u'Ru+(Ax\!+\!Bu\!+\!Gw)'P_f(Ax\!+\!Bu\!+\!Gw)\\
&\quad -(\lambda_-/2)\,w'w
\end{align*}
Since $\lambda_+\ge\lambda_-$, we have $q(u,w)\le r(u,w)$ for all $x,u,w$, hence
\[
\min_u\max_w q(u,w)\ \le\ \min_u\max_w r(u,w)
\]
Applying Proposition 14.a from Rawlings et al.\citep{rawlings:mannini:kuntz:2024,rawlings:mannini:kuntz:2024b} to both sides yields
$x'\Pi_1(\lambda_+)x \le x'\Pi_1(\lambda_-)x$ for all $x$, i.e., $\Pi_1(\lambda_+)\preceq\Pi_1(\lambda_-)$.
Therefore $m(\lambda)$ is continuous and nonincreasing on $[\lambda_2,\infty)$.

\begin{figure}
\centering
\includegraphics[width=1.02\linewidth]{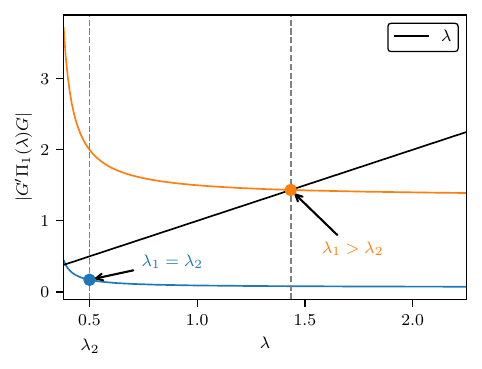}
\caption{Graphical construction of $\lambda_1$ via the fixed point of $\norm{G'\Pi_1(\lambda)G}$ (curves) against the identity $y=\lambda$ (solid black).
Dashed vertical lines mark $\lambda_2\eqbyd\norm{G'P_fG}=0.5$ and the resulting $\lambda_1$.
\emph{Blue:} $\lambda_2\ge \norm{G'\Pi_1(\lambda_2)G}$, so the admissibility condition $\lambda\ge \norm{G'\Pi_1(\lambda)G}$ already holds for all $\lambda\ge\lambda_2$ and we set $\lambda_1=\lambda_2$.
\emph{Orange:} $\lambda_2<\norm{G'\Pi_1(\lambda_2)G}$, so there is a unique $\lambda_1>\lambda_2$ with $\lambda_1=\norm{G'\Pi_1(\lambda_1)G}$; for all $\lambda\ge\lambda_1$ the condition $\lambda\ge \norm{G'\Pi_1(\lambda)G}$ holds.}
\label{fig:defer}
\end{figure}

\noindent
We thus distinguish two cases, illustrated in \cref{fig:defer}:
\begin{itemize}
\item If $\lambda_2\ge m(\lambda_2)$, then $\lambda\ge m(\lambda)$ for all $\lambda\ge\lambda_2$. Set $\lambda_1\eqbyd\lambda_2$.
\item If $\lambda_2<m(\lambda_2)$, define $n(\lambda)\eqbyd\lambda-m(\lambda)$. Then $n(\lambda)$ is continuous and strictly increasing on $[\lambda_2,\infty)$, with
$n(\lambda_2)<0$ and $\lim_{\lambda\to\infty}n(\lambda)=+\infty$.
By the intermediate value theorem there exists a unique $\lambda_1>\lambda_2$ such that $n(\lambda_1)=0$, i.e., $\lambda_1=m(\lambda_1)$.
\end{itemize}
In both cases, by construction $\lambda\ge m(\lambda)$ for all $\lambda\ge\lambda_1$. Combining the two cases, we define
\begin{equation*}
\lambda_1 \eqbyd \begin{cases}
\min_{\lambda\ge \lambda_2} \{\lambda: \lambda = \norm{G'\Pi_1(\lambda)G}\} & \text{if } \lambda_2 < \norm{G'\Pi_1(\lambda_2)G} \\
\lambda_2 & \text{if } \lambda_2 \geq \norm{G'\Pi_1(\lambda_2)G}
\end{cases}
\end{equation*}
Note that $\lambda_1\geq\lambda_2$.

Now that we have guaranteed the inequality $\lambda \geq \norm{G'\Pi_1(\lambda)G}$ holds for all $\lambda\geq\lambda_1$, from Proposition 14.a from Rawlings et al.\citep{rawlings:mannini:kuntz:2024,rawlings:mannini:kuntz:2024b} solutions to $\min_{u_0}\max_{w_0}  [\ell(x_0,u_0)+ x_1' \Pi_1(\lambda)x_1 - (\lambda/2)(w'_0w_0)]$ exist for $\lambda = \lambda_1$ for $d_0 \in \mathcal{R}(M_0(\lambda_1))$ and for $\lambda > \lambda_1$ for all $d_0 \in \bbR^{m+q}$. From \cref{prop:range-inclusion-invertible}, $G'\Pi_1(\lambda) G-\lambda I \preceq0$, and Assumptions 2-3, $M_0(\lambda)$ is invertible for $\lambda \geq\lambda_1$. Thus, $d_0 \in \mathcal{R}(M_0(\lambda_1))$ is always satisfied, and solutions exist for $\lambda \geq \lambda_1$ for all $d_0\in \bbR^{m+q}$. Furthermore, $\Pi_0(\lambda)$ is obtained from continuous operators, since $M_0(\lambda)$ is invertible, and well-defined for $\lambda \geq \lambda_1$, thus $\Pi_0(\lambda)$ is continuous for $\lambda \geq \lambda_1$. From the same arguments that proved $\Pi_1(\lambda)\succeq0$ for $\lambda\geq \lambda_2$, we have $\Pi_0(\lambda)\succeq0$ for $\lambda\geq \lambda_1$.

Define $\phi_0(\lambda,x_0)\eqbyd \min_{u_0}\max_{w_0}[\ell(x_0,u_0)+ (Ax_0+Bu_0+Gw_0)' \Pi_1(\lambda)(Ax_0+Bu_0+Gw_0) - (\lambda/2)w_0'w_0]$. We establish joint convexity in $(\lambda,x_0)$. For fixed $w_0$, the function $(u_0,\lambda,x_0) \mapsto \ell(x_0,u_0) + (Ax_0+Bu_0+Gw_0)'\Pi_1(\lambda)(Ax_0+Bu_0+Gw_0) - (\lambda/2)w_0'w_0$ is convex since $\phi_1(\lambda,x_1)$ is jointly convex in $(\lambda,x_1)$ and the composition with affine mapping $x_1=Ax_0+Bu_0+Gw_0$ preserves convexity \citep[§3.2.4]{boyd:vandenberghe:2004}. For fixed $(u_0,\lambda,x_0)$, the map $w_0 \mapsto \ell(x_0,u_0) + (Ax_0+Bu_0+Gw_0)'\Pi_1(\lambda)(Ax_0+Bu_0+Gw_0) - (\lambda/2)w_0'w_0$ is concave for $\lambda \geq \lambda_1$ since $G'\Pi_1(\lambda)G - \lambda I \preceq 0$. Therefore $g_0(u_0,\lambda,x_0)\eqbyd\max_{w_0}[\ell(x_0,u_0) + (Ax_0+Bu_0+Gw_0)'\Pi_1(\lambda)(Ax_0+Bu_0+Gw_0) - (\lambda/2)w_0'w_0]$ is the pointwise supremum of convex functions in $(u_0,\lambda,x_0)$, hence convex by \citep[§3.2.3]{boyd:vandenberghe:2004} (see also \citet[Theorem 5.5]{rockafellar:1970}). The partial minimization $\phi_0(\lambda,x_0) = \min_{u_0} g_0(u_0,\lambda,x_0)$ preserves joint convexity in $(\lambda,x_0)$ by \citep[§3.2.5]{boyd:vandenberghe:2004} (see also \citet[Theorem 5.3]{rockafellar:1970}).

\textbf{Third step: optimization over $\lambda$ at stage $k=0$.} \newline
Finally, the outer minimization over $\lambda\in [\lambda_1, \infty)$ yields
\[
 \min_{\lambda \in [\lambda_1, \infty)}
     \; \frac12\!\left(\frac{x_0}{\sqrt{\alpha}}\right)'\!
           \Pi_0(\lambda)\!\left(\frac{x_0}{\sqrt{\alpha}}\right)
     +\frac{\lambda}{2}
\]
The function 
$$L(\lambda) \eqbyd \frac12\left(\frac{x_0}{\sqrt{\alpha}}\right)'\Pi_0(\lambda)\left(\frac{x_0}{\sqrt{\alpha}}\right) + \frac{\lambda}{2}$$ 
is continuous on $[\lambda_1, \infty)$ because $\Pi_0(\lambda)$ is continuous for $\lambda \geq \lambda_1$. Moreover, $L(\lambda)$ is convex on $[\lambda_1,\infty)$ since $\phi_0(\lambda,x_0)$ is jointly convex in $(\lambda,x_0)$ as established in the second step. The function $L(\lambda)$ is coercive as $\lambda \to \infty$ since $\Pi_0(\lambda) \succeq 0$ implies $L(\lambda) \geq \lambda/2 \to \infty$. Therefore, by the Weierstrass theorem, a minimum exists with $\lambda^* \geq \lambda_1$.

\textbf{Completing the proof.} \newline
We finally prove items 1-4 from the proposition statement. Given the solution $\lambda^* \geq \lambda_1$ we have
\begin{enumerate}
    \item From the second and first step in the backward recursion and Proposition 14.a from Rawlings et al.\citep{rawlings:mannini:kuntz:2024,rawlings:mannini:kuntz:2024b} we have that the optimal solutions $u^*_0(x_0;\lambda^*)$ and $u^*_1(x_1;\lambda^*)$ satisfy \eqref{2dsc-a} and \eqref{2dsc-b}.
    \item From the second and first step in the backward recursion and Proposition 14.a from Rawlings et al.\citep{rawlings:mannini:kuntz:2024,rawlings:mannini:kuntz:2024b} we have that the solutions $\overline{w}_0(u^*_0)$ and $\overline{w}_1(u^*_1)$ satisfy \eqref{2dwcond-a} and \eqref{2dwcond-b}. Furthermore, the optimal solution $w^*_0(x_0;\lambda^*)$ and $w^*_1(x_1;\lambda^*)$ are jointly constrained within the set $\bbW$. Thus the optimal solutions satisfy $w^*_0(x_0;\lambda^*) = \overline{w}_0(u^*_0) \cap \bbW$ and $w^*_1(x_1;\lambda^*) = \overline{w}_1(u^*_1) \cap \bbW$.
    \item Given
\begin{align*}
L^*(\lambda^*) &= V^*(x_0, \useq^*, \wseq^*) - (\lambda^*/2) \left( (\wseq^*)'(\wseq^*) - \alpha \right) \\
&= (1/2)(\frac{x_0}{\sqrt{\alpha}})'\Pi_0(\lambda^*)(\frac{x_0}{\sqrt{\alpha}}) + \lambda^* /2
\end{align*}
and since $\wseq^*$ satisfies $(\wseq^*)'(\wseq^*) = \alpha$ from the constraint $\wseq \in \bbW$, we obtain
\[
V^*(x_0) = \min_{u_0}\max_{w_0} \; \min_{u_1}\max_{w_1} \;
 \frac{V(x_0, \useq, \wseq)}{\alpha} = L^*(\lambda^*)
\]
which is \eqref{2doc}.

\item
We now prove that for $\lambda\geq\lambda_1$, $\Pi_{0}(\lambda) \succeq \Pi_{1}(\lambda)\succeq P_f$, and $\Pi_0(\lambda)$, $\Pi_1(\lambda)$ are monotonic nonincreasing in $\lambda$. 

Fix $\lambda \geq \lambda_1$. First, from $\Pi_{1}(\lambda) = \bar{Q}_{1} + \bar{A}_{1}'P_f\bar{A}_{1} - \bar{A}_{1} P_f G(G'P_f G-\lambda I)^{\dagger}G'P_f \bar{A}_{1}$, $P_f \succeq 0$, and $\bar{A}_{1} P_f G(G'P_f G-\lambda I)^{\dagger}G'P_f \bar{A}_{1}\preceq0$, we have $\Pi_1(\lambda)\succeq P_f$. 

From $\Pi_{1}(\lambda) \succeq P_{f}$, we have $x'\Pi_{1}(\lambda) x \geq x'P_f x$ for all $x\in\bbR^n$. Define 
\begin{align*}
g(u,w) &= x'Qx + u'Ru \\
&\quad + (Ax+Bu+Gw)'\Pi_{1}(\lambda)(Ax+Bu+Gw) \\
&\quad - (\lambda/2)w'w\\
f(u,w) &= x'Qx + u'Ru \\
&\quad + (Ax+Bu+Gw)'P_f(Ax+Bu+Gw) \\
&\quad - (\lambda/2)w'w
\end{align*}
Since $\Pi_{1}(\lambda) \succeq P_f$, we have $g(u,w) \geq f(u,w)$ for all $x,u,w$, then $\min_u \max_w g(u,w) \geq \min_u \max_w f(u,w)$. Applying Proposition 14.a from Rawlings et al.\citep{rawlings:mannini:kuntz:2024,rawlings:mannini:kuntz:2024b} to both sides yields $x'\Pi_{0}(\lambda) x \geq x'\Pi_{1}(\lambda) x$ for all $x$, i.e., $\Pi_{0}(\lambda) \succeq \Pi_{1}(\lambda) \succeq P_f$.

Now consider $\lambda_{+} \geq \lambda_{-} \geq \lambda_1$. We previously proved $\Pi_{1}(\lambda_{+}) \preceq \Pi_{1}(\lambda_{-})$ for $\lambda \geq \lambda_2$. Because $\lambda_1 \geq \lambda_2$, then $\Pi_{1}(\lambda_{+}) \preceq \Pi_{1}(\lambda_{-})$ for $\lambda \geq \lambda_1$. Now we prove $\Pi_{0}(\lambda_{+}) \preceq \Pi_{0}(\lambda_{-})$ for $\lambda \geq \lambda_1$. Define
\begin{align*}
s(u,w) &= x'Qx + u'Ru \\
&\quad + (Ax+Bu+Gw)'\Pi_1(\lambda_{+})(Ax+Bu+Gw) \\
&\quad - (\lambda_{+}/2)w'w\\
t(u,w) &= x'Qx + u'Ru \\
&\quad + (Ax+Bu+Gw)'\Pi_1(\lambda_{-})(Ax+Bu+Gw) \\
&\quad - (\lambda_{-}/2)w'w
\end{align*}
Since $\lambda_{+} \geq \lambda_{-}$ and $\Pi_{1}(\lambda_{+}) \preceq \Pi_{1}(\lambda_{-})$, we have $s(u,w) \leq t(u,w)$ for all $x,u,w$, then $\min_u \max_w s(u,w) \leq \min_u \max_w t(u,w)$. Applying Proposition 14.a from Rawlings et al.\citep{rawlings:mannini:kuntz:2024,rawlings:mannini:kuntz:2024b} to both sides yields $x'\Pi_0(\lambda_+)x \le x'\Pi_0(\lambda_-)x$ for all $x$, i.e., $\Pi_0(\lambda_+)\preceq\Pi_0(\lambda_-)$.
\end{enumerate}
\end{proof}

\begin{proof}[Proof of \cref{prop:ndsignal}]
The proof follows by induction from \cref{prop:1dsignal}.

Define
\begin{align*}
M_k(\lambda) &\eqbyd \begin{bmatrix}
       B'\Pi_{k+1}B+R & B'\Pi_{k+1}G \\
       G'\Pi_{k+1}B & G'\Pi_{k+1}G-\lambda I
   \end{bmatrix} \\
d_k &\eqbyd \begin{bmatrix}
B'\Pi_{k+1}A \\ G'\Pi_{k+1}A
\end{bmatrix}x_k
\end{align*}
By arguments analogous to the stacked problem in \cref{prop:1dsignal}, the inequality constraint $|\mathbf{w}|^2 \leq \alpha$ can be replaced by the equality constraint $|\mathbf{w}|^2 = \alpha$, since \cref{asst3} ensures the maximum occurs on the boundary. Applying Proposition 7 from Rawlings et al.\citep{rawlings:mannini:kuntz:2024,rawlings:mannini:kuntz:2024b} to introduce the Lagrangian, \cref{th:minimax} to interchange $\max_{\wseq}$ and $\min_{\lambda}$ for $\lambda \geq \norm{\mathcal{G}'\mathcal{Q}\mathcal{G}}$, and Proposition 12.a in Rawlings et al. \citep{rawlings:mannini:kuntz:2024,rawlings:mannini:kuntz:2024b} to establish strong duality between $\min_{\useq}$ and $\max_{\wseq}$, we obtain
\begin{equation}
\begin{split}
V^*(x_0) \eqbyd  (1/\alpha) \min_{\lambda \geq \norm{\mathcal{G}'\mathcal{Q}\mathcal{G}}} \min_{u_0}\max_{w_0} &\min_{u_1}\max_{w_1} \\
& \cdots \min_{u_{N-1}}\max_{w_{N-1}} 
 L(x_0,\useq,\wseq,\lambda)
\end{split}
\label{signaldpL}
\end{equation}
where
\[
L(x_0,\useq,\wseq,\lambda) = \sum_{k=0}^{N-1} \ell(x_k,u_k) + \ell_f(x_N) - (\lambda/2)(\wseq'\wseq - \alpha)
\]
We apply backward dynamic programming to \eqref{signaldpL}, solving each minmax subproblem at stage $k$ using Proposition 14.a from Rawlings et al.\citep{rawlings:mannini:kuntz:2024,rawlings:mannini:kuntz:2024b} and determining the feasibility bound $\lambda_k$ at each stage as in \cref{prop:1dsignal}. At each stage $k$ we obtain
\begin{align*}
(1/\alpha)\min_{u_k}\max_{w_k} &[\ell(x_k,u_k) + x_{k+1}' \Pi_{k+1}(\lambda)x_{k+1} \\
&\quad - (\lambda/2)(w'_kw_k)] = (1/2\alpha)x_k'\Pi_k(\lambda)x_k
\end{align*}
and solutions exist for $\lambda \geq \lambda_k$ for all $d_k \in \mathbb{R}^{m+q}$. Moreover, by \cref{prop:range-inclusion-invertible}, $G'\Pi_{k+1}(\lambda)G - \lambda I \preceq0$, and Assumptions 2-3, the block matrix $M_k(\lambda)$ is invertible for every $\lambda\ge \lambda_k$, so the inverse in \eqref{signalrec1} is well-defined and $\Pi_k(\lambda)$ is continuous on $[\lambda_k,\infty)$. By induction, we obtain \eqref{ndsc2}, \eqref{ndwcond2}, \eqref{ndoc2}, the recursion \eqref{signalrec1} for $k \in [0,1,\dots,N-1]$ with terminal condition $\Pi_N = P_f$, and $\lambda_k \geq \lambda_{k+1}$. 

Define $\phi_k(\lambda,x_k)\eqbyd \min_{u_k}\max_{w_k}[\ell(x_k,u_k) + (Ax_k+Bu_k+Gw_k)' \Pi_{k+1}(\lambda)(Ax_k+Bu_k+Gw_k) - (\lambda/2)w_k'w_k]$. By arguments identical to those in \cref{prop:1dsignal}, $\phi_k(\lambda,x_k)$ is jointly convex in $(\lambda,x_k)$ for $\lambda \geq \lambda_k$.

The remaining optimization is
\[
\min_{\lambda \in [\lambda_1, \infty)}
    \; \frac12\!\left(\frac{x_0}{\sqrt{\alpha}}\right)'\!
          \Pi_0(\lambda)\!\left(\frac{x_0}{\sqrt{\alpha}}\right)
    +\frac{\lambda}{2}
\]
where $$L(\lambda)\eqbyd  \frac12\!\left(\frac{x_0}{\sqrt{\alpha}}\right)'\!
          \Pi_0(\lambda)\!\left(\frac{x_0}{\sqrt{\alpha}}\right)
    +\frac{\lambda}{2}$$
The function $L(\lambda)$ is continuous on $[\lambda_1, \infty)$ because $\Pi_0(\lambda)$ is continuous for $\lambda \geq \lambda_1$. Moreover, $L(\lambda)$ is convex on $[\lambda_1,\infty)$ since $\phi_0(\lambda,x_0)$ is jointly convex in $(\lambda,x_0)$. The function $L(\lambda)$ is coercive as $\lambda \to \infty$ since $\Pi_0(\lambda) \succeq 0$ implies $L(\lambda) \geq \lambda/2 \to \infty$. Therefore, by the Weierstrass theorem, a minimum exists with $\lambda^* \geq \lambda_1$ for all $d_k \in \bbR^{m+q}$.

Monotonicity properties follow by induction from extending the monotonicity arguments in \cref{prop:1dsignal} to finite horizon $N$.
\end{proof}

\begin{proof}[Proof of \cref{prop:derivative}]
We structure the proof in three parts: (i) establishing the equivalence of stacked and recursive problems, (ii) deriving the derivative formula, and (iii) determining the optimum.

\textbf{Equivalence via analytic continuation.} \newline
Consider the stacked Lagrangian optimization
\[
\tilde{L}(\lambda) \eqbyd (1/\alpha) \min_{\mathbf{u}} \max_{\mathbf{w}} \left[ V(x_0, \mathbf{u}, \mathbf{w}) - (\lambda/2)(|\mathbf{w}|^2 - \alpha) \right]
\]
Define the stacked matrix
\begin{equation}
\mathcal{M}(\lambda) \eqbyd \begin{bmatrix} \mathcal{B}'\mathcal{Q}\mathcal{B} + \mathcal{R} & \mathcal{B}'\mathcal{Q}\mathcal{G} \\ (\mathcal{B}'\mathcal{Q}\mathcal{G})' & \mathcal{G}'\mathcal{Q}\mathcal{G} - \lambda I \end{bmatrix}
\label{eq:calM-def}
\end{equation}
From Proposition 14.a from Rawlings et al.\citep{rawlings:mannini:kuntz:2024,rawlings:mannini:kuntz:2024b}, for $\lambda > \norm{\mathcal{G}'\mathcal{Q}\mathcal{G}}$ the saddle point of the stacked Lagrangian exists and the optimal value is
\[
\tilde{L}(\lambda) = \frac{1}{2\alpha} x_0' \Psi(\lambda) x_0 + \frac{\lambda}{2}
\]
where
\[
\Psi(\lambda) \eqbyd Q + \mathcal{A}'\mathcal{Q}\mathcal{A} - \mathcal{A}'\mathcal{Q} \begin{bmatrix} \mathcal{B} & \mathcal{G} \end{bmatrix} \mathcal{M}(\lambda)^{-1} \begin{bmatrix} \mathcal{B}' \\ \mathcal{G}' \end{bmatrix} \mathcal{Q}\mathcal{A}
\]
Similarly, consider the recursive problem
\begin{align*}
L(\lambda) \eqbyd (1/\alpha) \min_{u_0}\max_{w_0} &\cdots \min_{u_{N-1}}\max_{w_{N-1}} \\
&\quad \left[ V(x_0, \mathbf{u}, \mathbf{w}) - (\lambda/2)(|\mathbf{w}|^2 - \alpha) \right]
\end{align*}
From the proof of \cref{prop:ndsignal}, for $\lambda \geq \norm{\mathcal{G}'\mathcal{Q}\mathcal{G}}$, strong duality holds between $\min_{\mathbf{u}}$ and $\max_{\mathbf{w}}$ in the stacked problem, enabling equivalence with the recursive problem. Therefore, on the open set $\mathcal{S}^\circ \eqbyd (\norm{\mathcal{G}'\mathcal{Q}\mathcal{G}}, \infty)$, both problems yield the same optimal value: $L(\lambda) = \tilde{L}(\lambda)$ for all $\lambda \in \mathcal{S}^\circ$.

Define the recursive domain $\mathcal{D} \eqbyd [\lambda_1, \infty)$ where $\lambda_1$ is the feasibility bound from \cref{prop:ndsignal}. From \cref{prop:range-inclusion-invertible} and Assumptions 2-3, the matrices $M_k(\lambda)$ are invertible for all $\lambda \in \mathcal{D}$. The intersection $\mathcal{S}^\circ \cap \mathcal{D}$ is nonempty (it contains all sufficiently large $\lambda$), and on this set both problems are well-defined and yield the same optimal value $L(\lambda) = \tilde{L}(\lambda)$. We now establish that $\Pi_k(\lambda)$, $M_k(\lambda)^{-1}$, and $L(\lambda)$ are rational functions of $\lambda$. A scalar function $r(\lambda) = p(\lambda)/q(\lambda)$ with polynomials $p, q$ (and $q \neq 0$) is rational; a matrix-valued function is rational if each entry is a rational scalar function. Proceeding by induction on $k$:
\begin{itemize}
\item Base case: $\Pi_N(\lambda) = P_f$ is constant, hence polynomial, hence rational.
\item Inductive step: assume $\Pi_{k+1}(\lambda)$ has rational entries. The blocks $B'\Pi_{k+1}(\lambda)B + R$, $B'\Pi_{k+1}(\lambda)G$, $G'\Pi_{k+1}(\lambda)B$, and $G'\Pi_{k+1}(\lambda)G$ are sums and products of matrices with rational entries, hence have rational entries. Since $\lambda$ is a polynomial in $\lambda$, the $(2,2)$ block $G'\Pi_{k+1}(\lambda)G - \lambda I$ has rational entries. Therefore $M_k(\lambda)$ has rational entries. By the adjugate formula, $M_k(\lambda)^{-1} = \mathrm{adj}(M_k(\lambda))/\det M_k(\lambda)$; since $\mathrm{adj}(M_k(\lambda))$ and $\det M_k(\lambda)$ are polynomial expressions in the entries of $M_k(\lambda)$, both are rational in $\lambda$, hence $M_k(\lambda)^{-1}$ has rational entries. The Riccati recursion \eqref{signalrec1} involves only sums and products of matrices with rational entries, so $\Pi_k(\lambda)$ has rational entries.
\end{itemize}
By induction, $\Pi_0(\lambda)$ and $\tilde{J}(\lambda)$ have rational entries in $\lambda$. Since $L(\lambda)$ is a quadratic form in $x_0$ with coefficients from $\Pi_0(\lambda)$ plus the linear term $\lambda/2$, $L(\lambda)$ is rational in $\lambda$. Since $M_k(\lambda)$ is invertible for all $\lambda \in \mathcal{D}$, these rational functions have no poles on $\mathcal{D}$, hence are real-analytic on $(\lambda_1, \infty)$.

Similarly, $\mathcal{M}(\lambda) = \mathcal{M}_0 - \lambda \mathcal{I}$ where $\mathcal{M}_0$ is independent of $\lambda$ and $\mathcal{I} \eqbyd \mathrm{diag}(0, I)$, so $\Psi(\lambda)$ and $\tilde{L}(\lambda)$ are rational functions of $\lambda$, real-analytic on $\mathcal{S}^\circ$.

Since $L(\lambda) = \tilde{L}(\lambda)$ on the nonempty open set $\mathcal{S}^\circ \cap \mathcal{D}$, and both are real-analytic (being rational functions with no poles on their respective domains), by the identity theorem for real-analytic functions they are identical wherever both are defined.

\textbf{Derivative formula and monotonicity.} \newline
We prove items 1-2. Define the Lagrangian
\[
\mathcal{L}(\mathbf{u}, \mathbf{w}, \lambda) \eqbyd V(x_0, \mathbf{u}, \mathbf{w}) - (\lambda/2)(|\mathbf{w}|^2 - \alpha)
\]
For $\lambda \in \mathcal{S}^\circ$, from Proposition 14.a from Rawlings et al.\citep{rawlings:mannini:kuntz:2024,rawlings:mannini:kuntz:2024b} the stacked problem $\min_{\mathbf{u}} \max_{\mathbf{w}} \mathcal{L}(\mathbf{u}, \mathbf{w}, \lambda)$ admits a stationary point $\begin{bmatrix} \mathbf{u} \\ \mathbf{z} \end{bmatrix}^*(\lambda)$ satisfying
\begin{equation}
\mathcal{M}(\lambda) \begin{bmatrix} \mathbf{u} \\ \mathbf{z} \end{bmatrix}^*(\lambda) = -\begin{bmatrix} \mathcal{B}' \\ \mathcal{G}' \end{bmatrix} \mathcal{Q}\mathcal{A} \, x_0
\label{eq:stacked-stationary}
\end{equation}
The Lagrangian stationary point $\mathbf{z}^*(\lambda)$ can also be computed via the recursive gains. Define
\begin{align*}
K_k(\lambda) &\eqbyd -\begin{bmatrix} I & 0 \end{bmatrix} M_k(\lambda)^{-1} d_k(\lambda) \\
J_k(\lambda) &\eqbyd -\begin{bmatrix} 0 & I \end{bmatrix} M_k(\lambda)^{-1} d_k(\lambda)\\
F_k(\lambda) &\eqbyd A + BK_k(\lambda) + GJ_k(\lambda) \\
\Phi_{k,j}(\lambda) &\eqbyd F_{k-1}(\lambda)F_{k-2}(\lambda)\cdots F_j(\lambda)
\end{align*}
for $j < k$ and $\Phi_{k,k}(\lambda) \eqbyd I$, and
\[
\tilde{J}(\lambda) \eqbyd \begin{bmatrix}
J_0(\lambda) \\
J_1(\lambda) \Phi_{1,0}(\lambda) \\
\vdots \\
J_{N-1}(\lambda) \Phi_{N-1,0}(\lambda)
\end{bmatrix}
\]
From \cref{prop:ndsignal}, at each stage $k$ the stationary disturbance component satisfies $z_k^*(\lambda) = J_k(\lambda) x_k$ where $x_k = \Phi_{k,0}(\lambda) x_0$ is the state at stage $k$ under the closed-loop dynamics $x_{k+1} = F_k(\lambda) x_k$. Stacking yields $\mathbf{z}^*(\lambda) = \tilde{J}(\lambda) x_0$, which is well-defined on all of $\mathcal{D}$ since each $M_k(\lambda)$ is invertible there. On $\mathcal{S}^\circ$, the recursive formula coincides with the stacked expression \eqref{eq:stacked-stationary}, as both solve the same stationary conditions.

Define $F(\lambda) \eqbyd \mathcal{L}(\mathbf{u}^*(\lambda), \mathbf{z}^*(\lambda), \lambda)$. Differentiating with respect to $\lambda$ using the chain rule we obtain
\[
\frac{dF}{d\lambda} = \frac{\partial \mathcal{L}}{\partial \lambda} + \frac{\partial \mathcal{L}}{\partial \mathbf{u}} \frac{d\mathbf{u}^*}{d\lambda} + \frac{\partial \mathcal{L}}{\partial \mathbf{w}} \frac{d\mathbf{z}^*}{d\lambda}
\]
all evaluated at $(\mathbf{u}^*(\lambda), \mathbf{z}^*(\lambda), \lambda)$. At the stationary point, we have
\[
\frac{\partial \mathcal{L}}{\partial \mathbf{u}}\bigg|_{(\mathbf{u}^*, \mathbf{z}^*, \lambda)} = 0 \qquad \frac{\partial \mathcal{L}}{\partial \mathbf{w}}\bigg|_{(\mathbf{u}^*, \mathbf{z}^*, \lambda)} = 0
\]
Hence the last two terms vanish and we have
\[
\frac{dF}{d\lambda} = \frac{\partial \mathcal{L}}{\partial \lambda}\bigg|_{(\mathbf{u}^*, \mathbf{z}^*, \lambda)}
\]
Computing the partial derivative of $\mathcal{L}$ with respect to $\lambda$, we obtain
\[
\frac{\partial \mathcal{L}}{\partial \lambda}(\mathbf{u}, \mathbf{w}, \lambda) = -\frac{1}{2}(|\mathbf{w}|^2 - \alpha)
\]
Evaluating at the stationary point $(\mathbf{u}^*(\lambda), \mathbf{z}^*(\lambda), \lambda)$ we have
\[
\frac{dF}{d\lambda} = -\frac{1}{2}(|\mathbf{z}^*(\lambda)|^2 - \alpha) = \frac{1}{2}(\alpha - |\mathbf{z}^*(\lambda)|^2)
\]
Therefore, on $\mathcal{S}^\circ$, we obtain
\[
\frac{dL}{d\lambda} = \frac{1}{\alpha}\frac{dF}{d\lambda} = \frac{1}{2} - \frac{1}{2}\frac{|\mathbf{z}^*(\lambda)|^2}{\alpha}
\]
Since both sides are rational functions of $\lambda$ that agree on $\mathcal{S}^\circ$, by the identity theorem \eqref{eq:dL-dlambda} holds for all $\lambda \in \mathcal{D}$.

From \cref{prop:ndsignal}, $L(\lambda)$ is convex on $\mathcal{D}$. Since $L(\lambda)$ is real-analytic on $(\lambda_1, \infty)$, it is differentiable there, and by convex analysis $dL/d\lambda$ is nondecreasing on $(\lambda_1, \infty)$. From Part (i), $\tilde{J}(\lambda)$ has rational entries in $\lambda$ with no poles on $\mathcal{D}$, hence $dL/d\lambda$ is continuous on $\mathcal{D}$. Define $L'(\lambda_1) \eqbyd \lim_{\lambda \downarrow \lambda_1} dL/d\lambda$. Then $dL/d\lambda$ is nondecreasing on $[\lambda_1, \infty)$.

\textbf{Optimality conditions.} \newline
We prove item 3. Since $L(\lambda)$ is convex on $\mathcal{D}$ and coercive as $\lambda \to \infty$ (from $L(\lambda) \geq \lambda/2$), a minimizer $\lambda^* \in \mathcal{D}$ exists by the Weierstrass theorem. The minimizer satisfies $\lambda^* = \lambda_1$ if and only if $(dL/d\lambda)|_{\lambda = \lambda_1} \geq 0$.

From \eqref{eq:dL-dlambda}, we obtain
\[
\frac{dL}{d\lambda}\bigg|_{\lambda = \lambda_1} = \frac{1}{2}\left(1 - \frac{|\mathbf{z}^*(\lambda_1)|^2}{\alpha}\right) \geq 0 \quad \Longleftrightarrow \quad |\mathbf{z}^*(\lambda_1)|^2 \leq \alpha
\]
Therefore $\lambda^* = \lambda_1$ if and only if $|\mathbf{z}^*(\lambda_1)|^2 \leq \alpha$. Otherwise, $(dL/d\lambda)|_{\lambda_1} < 0$ and the minimizer satisfies $\lambda^* > \lambda_1$ with $(dL/d\lambda)|_{\lambda^*} = 0$, i.e., $|\mathbf{z}^*(\lambda^*)|^2 = \alpha$.
\end{proof}
\bibliographystyle{abbrvnat}
\bibliography{paper-ia_arxiv}

\end{document}